\documentclass[preprint,a4paper,10pt]{elsarticle}%

\usepackage[verbose=true,a4paper]{geometry}

\makeatletter
\def\ps@pprintTitle{%
 \let\@oddhead\@empty
 \let\@evenhead\@empty
 \def\@oddfoot{\centerline{\thepage}}%
 \let\@evenfoot\@oddfoot}
\makeatother

\usepackage{amsmath,amssymb}
\usepackage{amsthm} 
\usepackage{enumerate} 
\usepackage{algorithm}
\usepackage[noend]{algpseudocode}
\usepackage{hyperref} 
\usepackage{url,caption,subcaption}
\usepackage{color,bm}
\usepackage{soul} 
\usepackage{tikz}
\usepackage{makecell} 
\usepackage{array}
\newcolumntype{C}[1]{>{\centering\let\newline\\\arraybackslash\hspace{0pt}}m{#1}}
\newcolumntype{L}[1]{>{\raggedright\let\newline\\\arraybackslash\hspace{0pt}}m{#1}}
\newcolumntype{R}[1]{>{\raggedleft\let\newline\\\arraybackslash\hspace{0pt}}m{#1}}
\usepackage{booktabs} 
\usepackage{nicematrix}
\makeatletter
\def\@citex[#1]#2{%
  \let\@citea\@empty
  \@cite{\@for\@citeb:=#2\do
    {\@citea\def\@citea{;\penalty\@m\ }%
     \edef\@citeb{\expandafter\@firstofone\@citeb}%
     \if@filesw\immediate\write\@auxout{\string\citation{\@citeb}}\fi
     \@ifundefined{b@\@citeb}{\mbox{\reset@font\bfseries ?}%
       \G@refundefinedtrue
       \@latex@warning
         {Citation `\@citeb' on page \thepage \space undefined}}%
       {\csname b@\@citeb\endcsname}}}{#1}}
\makeatother
\usepackage{tabularx, environ}
\makeatletter
\newcolumntype{\expand}{}
\long\@namedef{NC@rewrite@\string\expand}{\expandafter\NC@find}
\NewEnviron{problem}[2][]{%
  \def\problem@arg{#1}%
  \def\problem@framed{framed}%
  \def\problem@lined{lined}%
  \def\problem@doublelined{doublelined}%
  \ifx\problem@arg\@empty%
    \def\problem@hline{}%
  \else%
    \ifx\problem@arg\problem@doublelined%
      \def\problem@hline{\hline\hline}%
    \else%
      \def\problem@hline{\hline}%
    \fi%
  \fi%
  \ifx\problem@arg\problem@framed%
    \def\problem@tablelayout{|>{\bfseries}lX|c}%
    \def\problem@title{\multicolumn{2}{|l|}{%
        \raisebox{-\fboxsep}{\textsc{\normalsize #2}}%
      }}%
  \else
    \def\problem@tablelayout{>{\bfseries}lXc}%
    \def\problem@title{\multicolumn{2}{l}{%
        \raisebox{-\fboxsep}{\textsc{\normalsize #2}}%
      }}%
  \fi%
  \bigskip\par\noindent%
  \begin{tabularx}{\textwidth}{\expand\problem@tablelayout}%
    \problem@hline%
    \problem@title\\[2\fboxsep]%
    \BODY\\\problem@hline%
  \end{tabularx}%
  \medskip\par%
}
\makeatother

\def\ra{{\rightarrow}}

\def\Lra{{\Leftrightarrow}}

\def\bF{{\mathbb F}}
\def\bN{{\mathbb N}}
\def\bZ{{\mathbb Z}}

\def\cC{{\mathcal C}}
\def\cGL{\mathrm{GL}}
\def\cI{{\mathcal I}}
\def\cJ{{\mathcal J}}

\def\cL{{\mathcal L}}
\def\cM{{\mathcal M}}

\def\cS{\mathrm{S}}

\def\cR{{\mathcal R}}
\def\cN{{\mathcal N}}
\def\cV{{\mathcal V}}

\newcommand{\rank}{\text{\rm rank}}
\newcommand{\Hw}[1]{\left|#1\right|}
\newcommand{\supp}[1]{\mathrm{Supp}(#1)}
\newcommand{\GI}{\mathcal{GI}}

\def\mA{\bm{A}}
\def\mB{\bm{B}}
\def\mC{\bm{C}}
\def\mD{\bm{D}}
\def\mG{\bm{G}}
\def\mO{\bm{0}}
\def\mP{\bm{P}}
\def\mQ{\bm{Q}}
\def\mS{\bm{S}}
\def\mX{\bm{X}}
\def\mY{\bm{Y}}
\def\mId{\bm{I}}
\def\mH{\bm{H}}
\def\mV{\bm{V}}
\def\mZ{\bm{Z}}
\def\m0{\bm{0}}

\def\vx{\bm{x}}
\def\vO{\bm{0}}
\def\vy{\bm{y}}
\def\vc{\bm{c}}

\def\vb{\bm{b}}
\def\vd{\bm{d}}
\def\ve{\bm{e}}
\def\vm{\bm{m}}
\def\vs{\bm{s}}
\def\vz{\bm{z}}
\def\vv{\bm{v}}

\def\v0{\bm{0}}

\newtheorem{remark}{Remark}
\newtheorem{example}{Example}
\newtheorem{proposition}{Proposition}
\newtheorem{lemma}{Lemma}
\newtheorem{definition}{Definition}

\newtheorem{theorem}{Theorem}
\newtheorem{corollary}{Corollary}
\newtheorem{cipher}{Cipher}

\begin{document}

\begin{frontmatter}

\title{Generalized Inverse Based Decoding}

\author[uaic]{Ferucio Lauren\c tiu \c Tiplea}
\ead{ferucio.tiplea@uaic.ro}
\author[uav,litis]{Vlad Florin Dr\u agoi}
\ead{vlad.dragoi@uav.ro}


\address[uaic]{Department of Computer Science, 
               ``Alexandru Ioan Cuza'' University of Ia\c si, Romania}

\address[uav]{Department of Computer Science, 
              Aurel Vlaicu University of Arad, Romania}

\address[litis]{LITIS Lab, Université de Rouen, 
                Avenue de l’Universitée, 
                76800 Saint-Étienne-du-Rouvray, France}

\begin{abstract}
The concept of Generalized Inverse based Decoding (GID) is introduced, as 
an algebraic framework for the syndrome decoding problem (SDP) and 
low weight codeword problem (LWP). 
The framework has ground on two characterizations by generalized 
inverses (GIs), one for the null space of a matrix and the other for the 
solution space of a system of linear equations over a finite field.
Generic GID solvers are proposed for SDP and LWP. 
It is shown that information set decoding (ISD) algorithms,
such as Prange, Lee-Brickell, Leon, and Stern's algorithms, are 
particular cases of GID solvers. 
All of them search GIs or elements of the null 
space under various specific strategies. 
However, as the paper shows the ISD variants 
do not search through the entire space, while our solvers do even 
when they use just one Gaussian elimination. 
Apart from these, our GID framework clearly shows how each ISD 
algorithm, except for Prange's solution, can be used as an SDP or LWP solver. 
A tight reduction from our problems, viewed as optimization problems, 
to the MIN-SAT problem is also provided. 
Experimental results show a very good behavior of the GID solvers. 
The domain of easy weights can be reached by a very few iterations 
and even enlarged. 

\end{abstract}


\end{frontmatter}

\tableofcontents

\section{Introduction}

\paragraph{McEliece's cryptosystem}

Rapid evolution of quantum computing \cite{LGBNA2021,Pres2018,TaFu2019}, 
as well as Shor's famous quantum algorithm for discrete logarithm 
and factorization \cite{Shor1994} are urging public key cryptography 
to move away from number theoretic based solutions. Initiated in 
2015 by the NIST, the post-quantum standardization process searches 
for quantum secure techniques for key exchange and digital signatures. 
Round 3 candidates \cite{PQC} 
for key exchange/key encapsulation mechanism are 
either code-based or lattice-based solutions. Amongst them, we 
discover one of the oldest public key encryption scheme, proposed 
in 1978 by McEliece \cite{McEl1978}. McEliece had the idea to generate 
a linear code that admits an efficient decoding algorithm 
(the private key) and to mask its structure (the masked code is 
the public key). Under the assumption that the public code is 
indistinguishable from a random code, breaking the scheme resumes 
to solving the syndrome decoding problem (SDP) for random codes, 
which is NP-complete \cite{BeET1978}. So, until one proves that 
there exist NP-complete problems that can be solved in polynomial 
time by quantum computers (such a result would be major breakthrough 
in complexity theory \cite{BeVa1997,Simo1994}) 
the McEliece scheme is considered quantum secure. 

As being the foundation of the McEliece scheme, a significant 
interest in finding algorithms for SDP emerged. Also, having 
a good estimation of the work factor required by algorithms for 
SDP is mandatory to asses the security level of the aforementioned scheme \cite{EsBe2021,HaSe2013,BBCPS2019,Pete2010}. 

\paragraph{Information set decoding}

Intiated in 1962 by E. Prange \cite{Pran1962}, Information Set Decoding 
(ISD) is a well-known technique for solving several fundamental problems 
in coding theory, e.g., SDP, low weight codeword 
problem (LWP), and even code equivalence problem (CEP) (see \cite{HuKS2021}). 
When $\bF_q=\bF_2$ both SDP and LWP are NP-Complete (see \cite{BeET1978} 
for SDP and \cite{Vard1997} for LWP). Up-to-date, for solving SDP and LWP, 
the latest variants of ISD \cite{BJMM2012,MaOz2015,BoMa2018} are the most 
efficient techniques, in term of 
time complexity. The decisional SDP takes as input the parity-check matrix 
$\mH\in\cM_{n-k,n}(\bF_q)$ of a linear code $\cC$ over a finite field 
with $q$ elements $\bF_q$, a syndrome vector $\vs\in\bF_q^{n-k}$ and an 
integer $t$, and asks if there is a solution to the equation $\mH\vx=\vs$ 
satisfying $\Hw{\vx}\leq t,$ where ``$\Hw{\vx}$'' denotes the Hamming weight of the 
vector $\vx$. In SDP, syndrome vectors are always non-zero. When $\vs=\mO$ 
the problem becomes LWP. The main idea behind the original ISD technique 
(Prange's approach) 
is to pick a sufficiently large set of error-free coordinates such that 
the corresponding columns of $\mH$ form an invertible submatrix. This is equivalent 
to computing two matrices $\mP\in\cGL_{n-k}(\bF_q)$ and $\mQ\in\cS_n(\bF_q)$ 
such that 
$\mP\mH\mQ=\begin{pmatrix} \mV & \mId_{n-k} \end{pmatrix}.$ This gives us
$\vx=\mQ\begin{pmatrix} \mO\\\mP\vs \end{pmatrix}.$ The correct information 
set for SDP is found when $\Hw{\mP\vs}\leq t.$ In other words, this procedure 
will stop eventually, as there is a permutation which sends the support of 
the solution outside the information set. 
Relaxation of certain architectural constraints of Prange's 
algorithm or the addition of optimizations, mainly based on the birthday 
paradox (also known as the meet-in-the-middle approach), has led to many 
improvements  
\cite{LeBr1988,Ster1988,Leon1988,Dume1989,Dume1991,CoGF1991,CaCh1998,FiSe2009,BeLP2011,MaMT2011,BJMM2012,MaOz2015,BoMa2018}.

The difficulty of solving SDP highly depends on the range of values for 
the parameter $t.$ In a cryptographic context, it is frequent to select hard 
instances, i.e., where $t$ is close to the Gilbert-Varshamov bound or 
sub-linear in the code-length \cite{BeET1978,Gilb1952,Vars1957}. 
When $t$ is linear in $n$, optimizations 
to Prange's algorithm have better complexity results even in the first term 
in the exponent. However, when $t$ is sub-linear in $n$, which is the case 
of all NIST code-based submissions, the advantage of all the improvements 
vanish asymptotically. To be more precise, for $t=o(n)$ and $n\to\infty$ 
the work factor of existing algorithms for SDP equals 
$2^{-t\log(1-k/n)(1+o(1))}$ \cite{CaSe2016}. 

It worths mentioning that ISD is not the only technique for solving SDP. 
For example, statistical decoding has a quite different approach 
\cite{FoKI2007,Jabr2001,Over2006,DeTJ2017}. However, it does not achieve
performance comparable to even the simplest ISD techniques, e.g., 
Prange's algorithm. 

Both SDP and LWP can be seen as particular cases of some well-known 
generic problems, Coset Weight Problem (CWP) and Subspace Weight Problem 
(SWP) \cite{BeET1978}. 
The difference between CWP and SDP, repectively between SWP and LWP, 
resides in the input matrix $\mA$ which for CWP and SWP is an arbitrary matrix 
from $\cM_{m,n}(\bF).$ Notice that when $m=n-k\leq n$ and $\rank(\mA)=n-k$, 
CWP becomes SDP. In order to stay as general as possible, in this 
paper we will present solutions for CWP/SWP and restrict to full 
rank matrices when discussions move towards coding theory.  

\paragraph{Generalized inverse of a matrix}

Since one has to solve a system of linear equations to find a solution 
for CWP, the idea of computing the inverse of the matrix $\mA$ comes 
natural in mind. However, as $\mA$ is not square, we can not apply 
this technique here.  
Nevertheless, the concept of matrix inverse exists in the case of 
non-square matrices. 
It is known as \textit{generalized inverse} (GI). Given a matrix 
$\mA\in\cM_{m,n}(\bF)$, a GI for $\mA$ is a matrix $\mX\in\cM_{n,m}(\bF)$ satisfying $\mA\mX\mA=\mA.$ Several types of 
inverses are known, such as, reflexive, normalized, and 
pseudo-inverse (or Moore-Penrose inverse \cite{Moor1920,Penr1955}). 
In particular, the Moore-Penrose inverse is a helpful tool 
when minimum norm solutions are required over the field of 
real or complex numbers. However, when moving to finite fields, 
things change a lot, mainly due to the geometrical properties 
of the scalar product. Results regarding the GI 
in arbitrary and finite fields exist 
\cite{Pear1968,Fult1978,BeGr2006} and are going to be used 
and extended here in the context of linear codes. 

There were several attempts to use GI of a 
matrix in cryptography and coding theory
\cite{WuDa1998,Sun2001,NgDa2013,NgDa2017,Fini2005}.
In 1998, Wu and Dawson \cite{WuDa1998} have proposed a public-key
cryptosystem based on GIs, but three years later 
it was cryptanalyzed \cite{Sun2001}. 
Dang and Nguyen \cite{NgDa2013,NgDa2017} have used 
pseudo-inverses, the strongest form of a GI, to
design key exchange protocols and protocols for privacy-preserving 
auditing data in cloud. 

The only reference of GIs with respect to SDP is by Finiasz \cite{Fini2005} 
($\mH$ is the parity-check matrix,
$\mS$ is the syndrome, and the threshold is $w$): 
\begin{quote}
``For instance, when $w$ is 
larger than $n/2$, solving SD becomes easy, as computing a 
pseudo-inverse $\mH^{-1}$ of $\mH$ and computing $\mH^{-1}\cdot\mS$ 
will return a valid solution with large probability. However, for 
smaller values of $w$, when a single solution exists, finding it 
becomes much harder.'' 
\end{quote}

The pseudo-inverse of a matrix, when it exists, is unique. Its use 
in enumerating a space of possible values is doomed to failure. 
But if the inverse concept is relaxed, we can broaden the search 
spectrum, and things become affordable. But the question is: 
how affordable? Can we list the entire possible solution space?
This is the question our paper wants to answer.

\paragraph{Contributions}
We discuss below the contributions that our work makes.

\smallskip
\textbf{GI based solvers for CWP and SWP.}
Our first contribution is to propose an algebraic formalism based on the 
GI of a matrix to address both CWP and SWP. This formalism 
allows us to have a unified vision of the two problems, and it provides the 
main tools for understanding all the algorithmic improvements for solving 
CWP and SWP. 

We begin by a careful inspection of the solutions of a linear system of 
equations $\mA\vx=\vb$ with $\vb\not=\mO$, and prove that all its solutions 
can be obtained only by GIs. More precisely, we show that 
\begin{equation*}
\nonumber
  \{\vx\in\bF_q^n\;|\;\mA\vx=\vb\}=\{\mX\vb\;|\;\mX \text{ is a GI of }\mA\}.
\end{equation*}
This characterization allows us to attack CWP in a very direct way:
sample GIs (by means of some strategy) until a solution with the desired 
Hamming weight is reached. 
For example, one could fix a transformation $(\mP,\mQ)\in \cGL_r(\bF)\times \cS_n(\bF)$
      with $\mP\mA\mQ=\begin{pmatrix} \mV & \mId_m\end{pmatrix}$ for some
      $\mV$, and then search solutions of the form $\mX\vb$, where  
    $$\mX\in \left\{\mQ\begin{pmatrix} \mZ \\ \mId_r-\mV\mZ\end{pmatrix}\mP 
           \mid \mZ\in\cM_{n-m,m}(\bF)\right\}.$$ 
This method covers the whole space of solutions of the system $\mA\vx=\vb$. 

For SWP, we prove first that the null space of $\mA$ can be
characterized by
\begin{equation*}
\nonumber
  \{\vx\in\bF_q^n\;|\;\mA\vx=\mO\}=\{(\mY-\mX)\vb\;|\;\mY \text{ is a GI of }\mA\},
\end{equation*}
where $\mX$ ($\vb$, resp.) is an arbitrary but fixed GI of $\mA$ (non-zero vector, resp.).
Thus, this characterization allows us to design a generic algorithm for SWP as 
the one above: fix first a GI $\mX$ of $\mA$ and a non-zero vector $\vb$, and 
then sample GIs $\mY$ of $\mA$ until a solution with the desired Hamming 
weight is reached. The sampling of GIs is with respect to some strategy. 
For instance, if we decompose $\mA$ into 
$\mP\mA\mQ=\begin{pmatrix} \mV & \mId_m \end{pmatrix}$, then the null
space of $\mA$ is 
\begin{equation*}
        \left\{\mQ\begin{pmatrix} \mZ \\ -\mV\mZ \end{pmatrix}\mP\vb
               \bigm| \mZ\in\cM_{n-m,m}(\bF)
        \right\}.
\end{equation*}      
So, the sampling can be on arbitrary matrices $\mZ$. 

We will use the terminology {\em GI based Decoding} (GID) 
to refer to any of the GI-based techniques presented above, and 
{\em GID solver} for any algorithm that falls under it.

\smallskip
\textbf{Information set decoding versus GID.} 
Our GID technique works as a common denominator for many existing 
information set decoding techniques, such as Prange, Lee-Brickell, Leon, 
Stern, Finiasz-Sendrier (and probably all). It explains the essence of 
all these methods in a very clear and unified way. 
For instance, we show in the paper that Prange's algorithm computes 
particular GIs until it finds the desired solution, without 
covering the entire space of solutions. 
More exactly, given $\mH\in\cM_{n-k,n}(\bF)$ a parity-check matrix of 
a linear code and a syndrome $\vs\in\bF^{n-k}$, Prange's algorithm generates 
solutions to the equation $\mH\vx=\vs$ of the form $\mX\vs$, where 
\begin{equation*}
    \mX\in\left\{\mQ\begin{pmatrix} \mO \\ \mId_r
                  \end{pmatrix}\mP \mid
              (\mP,\mQ)\in \cGL_r(\bF)\times \cS_n(\bF),\ 
              (\exists\mV:\,\mP\mH\mQ=\begin{pmatrix} 
                                          \mV&  \mId_r
                                       \end{pmatrix})\right\}.
\end{equation*}

However, as we prove in the paper, $\mP$ must be in $\cGL_m(\bF)$
to cover the entire space of solutions. The same holds for the other
ISD techniques discussed in  paper, and probably for all techniques
that share Prange's idea. 

In terms of GID, each ISD technique is just a strategy to search a 
partial subspace of the space of solutions or of the null space. This
view allows us to easily convert each such ISD technique into one 
working for SDP or LWP.

\smallskip
\textbf{A tight reduction to MIN-SAT.} 
Both CWP and SWP can be viewed as decision problems associated to two 
optimization problems, namely the {\em minimum coset weight problem} 
(MIN-CWP) and the {\em minimum subset weight problem} (MIN-SWP). 
Our GI-based approach allows to tightly reduce these optimization
problems (MIN-CWP and MIN-SWP) to the well-known MIN-SAT problem, when
$\bF=\bF_2$. 
As the reduction is very tight, we expect many techniques working for MIN-SAT 
to apply to the two problems.

\smallskip
\textbf{Reaching easy weights by means of GID.}
Our simulations have shown that for small length codes, GID solvers  
behave very similar to ISD decoders in terms of performance. 
We have also noticed through simulations that there is an interval of 
Hamming weights where the GID solvers manage to efficiently find solutions 
for CWP and SWP. For example 
when $\bF=\bF_2$ the interval is symmetric and centered in $n/2$ (see \cite{DAST2019}). 
Our simulations suggest that it is rather easy in general to find solutions 
within this range. However, we know that hard instances exist even for this 
interval (see for example \cite{Vard1997}). 
Our simulations show that for codes of length up to $n=3000$, with just  
one $(\mP,\mQ)$ decomposition, we have reached solutions with Hamming weights in the 
range $[r\frac{q-1}{q}-\sqrt{n},r\frac{q-1}{q}+n-r+\sqrt{n}]$ in only a few 
seconds on an ordinary laptop computer. 


\paragraph{Paper organization}

The article begins by setting the notation and basic definitions from 
coding theory (next section). Section \ref{GI} is dedicated to the
GI. The two central problems, CWP and SWP, are 
treated in Section \ref{CWGI}, where two generic GID solvers for them
are presented.  
Moving forward to ISD, Section \ref{ISD} starts with some historical 
considerations. After that, its focus is on the first ISD decoder, 
i.e., Prange's algorithm (\ref{ISD1}). Till the end of Section \ref{ISD} 
several variants of ISD are considered within the framework of GID. 
The GI allows us to make a closed reduction from CWP 
and SWP, viewed as optimization problems, to the well-known MIN-SAT 
problem (Section \ref{GID-CSP}). 
Section \ref{ExpRez} considers some practical issues, by providing 
experimental tests on a variety of code parameters.        

\section{Preliminaries}\label{sec:Prelim}

We fix the basic notation on linear algebra and coding 
theory that we will use in the paper (for details, the reader 
is referred to the standard textbooks such as 
\cite{Roma2007,Gent2017,Roth2006,HuKS2021}).

Generic fields are denoted by $\bF$. When we want to emphasize 
that a field is finite and has the order $q$, we will write 
$\bF_q$. $\bF^n$ stands for the $n$-dimensional vector space 
over $\bF$. The vectors of $\bF^n$ will be denoted by lowercase 
letters, such as $\vx$, and written in column form. 
The $i$th element of $\vx\in\bF^n$ is denoted $\vx(i)$, where $1\leq i\leq n$,
and the support of $\vx$ is  
$\supp{\vx}=\{i\mid 1\leq i\leq n,\,\vx(i)\not =0\}$.
The cardinality of $\supp{\vx}$ is the Hamming weight of $\vx$.
We shall simply denote this as $\Hw{\vx}$. 
If $\cI$ is a non-empty subset of $\{1,\ldots,n\}$, $\vx_{\cI}$ stands 
generally for the restriction of $\vx$ to $\cI$, that is, the vector 
of size $|\cI|$ (the cardinality of $\cI$) that is obtained
from $\vx$ by removing all entries on positions outside $\cI$. 
The operator ``$|\cdot|$'' is used both for the Hamming weight and the 
cardinality of a set. However, the distinction will always be clear 
from the context.  

The set of $m\times n$ matrices with elements in $\bF$ is denoted 
$\cM_{m,n}(\bF)$. Matrices will be denoted by uppercase bold 
letters, such as $\mA$. $\mA(i,j)$ denotes the element of $\mA$ at 
the intersection of row $i$ and column $j$. 
$\mId_r$ stands for the identity matrix of size $r$, and 
$\mId_{m,n,r}$ is $\mId_r$ extended with zeroes to an 
$m\times n$ matrix, i.e., 
$\mId_{m,n,r}=\begin{pmatrix} \mId_r & \mO \\ \mO & \mO \end{pmatrix}$. 
$\cGL_m(\bF)$ denotes the general linear group of order $m$
over $\bF$ (i.e., the group of all invertible matrices 
$\mA\in\cM_{m,m}(\bF)$). Its subgroup consisting of 
permutation matrices is denoted $\cS_n(\bF)$. 
$\pi_{\mQ}$ stands for the permutation induced by 
$\mQ\in\cS_n(\bF)$, and $\pi([i,j])$ is the image of the interval 
$[i,j]$ through the permutation $\pi$. 

As usual, $\mA^t$ ($\mA^{-1}$, $\rank(\mA)$) 
stands for the transpose (inverse, rank) of $\mA\in\cM_{m,n}(\bF)$. 
The {\em range} ({\em null space}) of $\mA$ is 
$\cR(\mA)=\{y\in\bF^n \mid \exists\vx\in\bF^m:\,\vx^t\mA=\vy^t\}$
($\cN(\mA)=\{\vx\in\bF^n \mid \mA\vx=\vO\}$). 
We will use $\langle\mA\rangle$ to denote the vector space spanned by $\mA.$

A linear $[n,k]$ code over $\bF_q$ is a vector subspace $\cC$ of 
$\bF_q^n$ of dimension $k$. 
Any matrix $\mG\in\cM_{k,n}(\bF_q)$ whose rows form a basis 
for $\cC$ is a generator matrix for $\cC$. 
A parity-check matrix for $\cC$ is a generator matrix $\mH$ 
for the dual code $\cC^\perp$.

\section{The GI of a matrix over arbitrary fields}
\label{GI}

The GI of a matrix has been much studied over the 
fields of real and complex numbers. Not all the results valid in this 
context remain valid when moving to an arbitrary field, especially 
to finite fields. Consequently, in this section, we will recall some 
results that are valid for matrices over arbitrary fields, and when 
needed, we will specialize them to finite fields. 
We will mainly follow \cite{BeGr2006,Pear1968,Fult1978}, but we draw 
attention to the fact that some results will be presented in our own 
approach, which we consider appropriate for the case of finite fields.

\subsection{Definitions and existence}\label{GI1}

\begin{definition}
Let $\bF$ be an arbitrary field and $\bm{A}\in\cM_{m,n}(\bF)$ 
be a matrix. A {\em GI of $\mA$} is any matrix 
$\mX\in\cM_{n,m}(\bF)$ that fulfills 
\begin{equation}\label{GI-Eq01}
\mA\mX\mA=\mA.
\end{equation}
\end{definition}

There are specialized cases of GI, but we do not 
mention them here because they are not used in our paper.

From the definition one can easily see that $\mA^{-1}$ is the only 
GI of $\mA$ when $\mA$ is non-singular. 
That is, in such a case, the GI exists and is unique. 
Before moving on to the analysis of the existence of the generalized 
inverse in the general case, let us analyze in more detail its 
definition and the connection with solving systems of linear 
equations. Let $\GI(\mA)$ stand for the set of GIs 
of $\mA$.

\begin{theorem}[\cite{BeGr2006}]\label{GI-T01}
Let $\mA\in\cM_{m,n}(\bF)$ and $\mX\in\cM_{n,m}(\bF)$. Then, 
  $$\mX\in \GI(\mA)\ \ \Lra\ \ 
    (\forall \vb\in \cR(\mA))(\mX\vb \text{ is a solution to } \mA\vx=\vb).$$ 
\end{theorem}

\begin{theorem}[\cite{BeGr2006}]\label{GI-T02}
Let $\mA\in\cM_{m,n}(\bF)$, $\vb\in \cR(\mA)$, and $\mX\in \GI(\mA)$. 
Then, $\vx$ is a solution to $\mA\vx=\vb$ if and only if 
$\vx=\mX\vb + (\mId-\mX\mA)\vc$, for some $\vc\in \bF^n$.  
\end{theorem}

It is also well known that $\cR(\mId-\mX\mA)=\cN(\mA)$
(see, for instance, \cite{BeGr2006,Rohd2003}). Therefore, 
by Theorem \ref{GI-T02}, $\mX\vb$ is a solution to the system $\mA\vx=\vb$
(when it is consistent) and any other solution can be obtained
by adding arbitrary elements from the null space of $\mA$ to $\mX\vb$. 

GIs exist for all matrices over arbitrary fields
\cite{Pear1968,Fult1978}. A first step in showing this is based on the 
following theorem.

\begin{theorem}[\cite{BeGr2006}]\label{GI-T03}
Let $\mA\in\cM_{m,n}(\bF)$, $\mP\in\cGL_m(\bF)$,
and $\mQ\in\cGL_n(\bF)$. Then, the function
$f:\GI(\mA)\ra \GI(\mP\mA\mQ)$ given by $f(\mX)=\mQ^{-1}\mX\mP^{-1}$, 
for any $\mX\in \GI(\mA)$, is a bijection. 
\end{theorem}
\begin{proof}
It is straightforward to check that $f$ is well-defined and 
one-to-one. It remains to prove that any GI $\mY$ 
of $\mP\mA\mQ$ is of the form $\mQ^{-1}\mX\mP^{-1}$, for some 
$\mX\in \GI(\mA)$. 

If $\mY\in \GI(\mP\mA\mQ)$, then $\mP\mA\mQ\mY\mP\mA\mQ=\mP\mA\mQ$, 
which is equivalent to $\mA\mQ\mY\mP\mA=\mA$, since $\mP$ and $\mQ$ 
are non-singular matrices. However, this shows that 
$\mX=\mQ\mY\mP\in \GI(\mA)$ and $\mY=\mQ^{-1}\mX\mP^{-1}$.
\end{proof}

\begin{corollary}\label{GI-C01}
Let $\mA,\mB\in\cM_{m,n}(\bF)$. If $\mP\mA\mQ=\mB$ for some 
matrices $\mP\in\cGL_{m}(\bF)$ and $\mQ\in\cGL_{n}(\bF)$, 
then:
\begin{enumerate}
\item $\GI(\mA)=\{\mQ\mX\mP\mid \mX\in \GI(\mB)\}$;
\item $|\GI(\mA)|=|\GI(\mB)|$. 
\end{enumerate}
\end{corollary}
\begin{proof}
Apply Theorem \ref{GI-T03} to $\mB$ and $\mA=\mP^{-1}\mB\mQ^{-1}$. 
\end{proof}

\subsection{Computing GIs}\label{GI2}

To facilitate expression, a pair $(\mP,\mQ)$ of matrices as in 
Corollary \ref{GI-C01} will often be called a {\em transformation} of 
$\mA$. The first part of this corollary shows that the set 
$\GI(\mA)$ does not depend on the transformation we apply to $\mA$.
As a result, it suggests the following method for computing  
GIs of $\mA$:
\begin{itemize}
\item Transform the matrix $\mA$ through elementary operator
  matrices $\mP$ and $\mQ$ into a matrix $\mB=\mP\mA\mQ$ for which one 
  can easily compute GIs;
\item For each GI $\mX$ of $\mB$, $\mQ\mX\mP$ is a 
  GI of $\mA$. Besides, all GIs of 
  $\mA$ are obtained in this way.
\end{itemize}

As an example, one may use the canonical form of $\mA$, 
$\mP\mA\mQ=\mId_{m,n,r}$ \cite{Gent2017}. Thus, computing generalized 
inverses for $\mA$ is reduced to computing GIs for 
$\mId_{m,n,r}$. We present below some general constructions that 
also include this case. 
Even if the results are trivial to prove, we prefer to present them 
in the form of a proposition to highlight their usefulness further.

\begin{proposition}\label{GI-P01}
Let $\mA\in\cM_{m,n}(\bF)$ and $\mX\in\cM_{n,m}(\bF)$ be matrices.
\begin{enumerate}
\item If $\mA$ and $\mX$ are divided into blocks of appropriate sizes,
  $\mA=\begin{pmatrix} \mA_1 & \mA_2 \\ \mA_3 & \mA_4 \end{pmatrix}$ and
  $\mX=\begin{pmatrix} \mX_1 & \mX_2 \\ \mX_3 & \mX_4 \end{pmatrix}$, 
  respectively, then $\mX$ is a GI of $\mA$ if and
  only if the following matrix equations are fulfilled:
  \begin{equation}\label{GI-Eq02}
  \left\{
  \begin{aligned}
  (\mA_1\mX_1+\mA_2\mX_3)\mA_1+(\mA_1\mX_2+\mA_2\mX_4)\mA_3=\mA_1 \\ 
  (\mA_1\mX_1+\mA_2\mX_3)\mA_2+(\mA_1\mX_2+\mA_2\mX_4)\mA_4=\mA_2 \\
  (\mA_3\mX_1+\mA_4\mX_3)\mA_1+(\mA_3\mX_2+\mA_4\mX_4)\mA_3=\mA_3 \\
  (\mA_3\mX_1+\mA_4\mX_3)\mA_2+(\mA_3\mX_2+\mA_4\mX_4)\mA_4=\mA_4 
  \end{aligned}
  \right.
  \end{equation}
\item If $\mA$ and $\mX$ are divided into blocks of appropriate sizes,
  $\mA=\begin{pmatrix} \mA_1 & \mA_2 \end{pmatrix}$ and 
  $\mX=\begin{pmatrix} \mX_1 \\ \mX_2 \end{pmatrix}$, respectively, 
  then $\mX$ is a GI of $\mA$ if and
  only if the following matrix equations are fulfilled:
  \begin{equation}\label{GI-Eq03}
  \left\{
  \begin{aligned}
  (\mA_1\mX_1+\mA_2\mX_2)\mA_1=\mA_1 \\
  (\mA_1\mX_1+\mA_2\mX_2)\mA_2=\mA_2
  \end{aligned}
  \right.
  \end{equation}
\item If $\mA$ and $\mX$ are divided into blocks of appropriate sizes,
  $\mA=\begin{pmatrix} \mA_1 \\ \mA_2 \end{pmatrix}$ and
  $\mX=\begin{pmatrix} \mX_1 & \mX_2 \end{pmatrix}$, respectively, 
  then $\mX$ is a GI of $\mA$ if and
  only if the following matrix equations are fulfilled:
  \begin{equation}\label{GI-Eq04}
  \left\{
  \begin{aligned}
  \mA_1(\mX_1\mA_1+\mX_2\mA_2)=\mA_1 \\
  \mA_2(\mX_1\mA_1+\mX_2\mA_2)=\mA_2
  \end{aligned}
  \right.
  \end{equation}
\end{enumerate}
\end{proposition}
\begin{proof}
Directly from \eqref{GI-Eq01}. 
\end{proof}

\begin{example}\label{GI-E01}
We present below some cases of application of Proposition \ref{GI-P01}. 
\begin{enumerate}
\item If $\mA=\begin{pmatrix} \mId_r & \mO \\ \mO & \mO \end{pmatrix}$, 
  then the GIs of $\mA$ have the form 
  $\mX=\begin{pmatrix} \mId_r & \mX_2 \\ \mX_3 & \mX_4 \end{pmatrix}$, where 
  $\mX_2$, $\mX_3$, and $\mX_4$ are arbitrary matrices (of appropriate sizes)
  over $\bF$.
  When $\bF=\bF_q$ we have $|\GI(\mA)|=q^{mn-r^2}$ (see \cite{Fult1978});
\item If $\mA=\begin{pmatrix} \mId_r & \mA_2 \\ \mO & \mO \end{pmatrix}$, 
  then the GIs of $\mA$ have the form 
  $\mX=\begin{pmatrix} \mX_1 & \mX_2 \\ \mX_3 & \mX_4 \end{pmatrix}$, 
  where $\mX_1$, $\mX_2$, $\mX_3$, and $\mX_4$ are arbitrary matrices 
  (of appropriate sizes) over $\bF$ that satisfy $\mX_1+\mA_2\mX_3=\mId_r$.
\item If $\mA=\begin{pmatrix} \mA_1 & \mId_r \end{pmatrix}$, 
  then the GIs of $\mA$ have the form 
  $\mX=\begin{pmatrix} \mX_1 \\ \mX_2 \end{pmatrix}$, 
  where $\mX_1$ and $\mX_2$ are arbitrary matrices 
  (of appropriate sizes) over $\bF$ that satisfy $\mA_1\mX_1+\mX_2=\mId_r$.
\item If $\mA=\begin{pmatrix}\mId_r & \mO\end{pmatrix}$, 
  then the GIs of $\mA$ have the form  
  $\mX=\begin{pmatrix} \mId_r \\ \mX_2 \end{pmatrix}$, 
  where $\mX_2$ is a matrix (of appropriate size) over $\bF$. 
\item If $\mA=\begin{pmatrix} \mId_r \\ \mO\end{pmatrix}$,
  then the GIs of $\mA$ have the form  
  $\mX=\begin{pmatrix} \mId_r & \mX_2 \end{pmatrix}$, where 
  $\mX_2$ is a matrix (of appropriate size) over $\bF$. 
\end{enumerate}
\end{example}

Even if the set of GIs of a matrix $\mA$ is the 
same regardless of the transformation applied to the matrix, 
the equations that define them may be different. 
For instance, a matrix of rank $r$ can be transformed into the
form in Example \ref{GI-E01}(1) and in the form in 
Example \ref{GI-E01}(2). 
The equations that define the GI will be different, 
even if, in the end, we get the same set of GIs. 
But, as will be seen later, these equations will more or less 
facilitate working with GIs (enumeration, processing, 
convergence to solution).
Due to this, given a transformation
$(\mP,\mQ)$ for $\mA$, we will denote by $\GI_{\mP,\mQ}(\mA)$ 
the set of GIs of $\mA$ obtained from the 
GIs of $\mP\mA\mQ$. It is clear that
$\GI_{\mP,\mQ}(\mA)=\GI(\mA)$.

\section{The coset and subspace weight problems and the GI}
\label{CWGI}

This section will discuss possible applications of the generalized 
inverse in solving two closely related hard problems in coding theory. 
The first of them is the {\em coset weight} problem, and the 
second is the {\em subspace weight} problem, a subproblem of the 
first one.

\subsection{Generic GID solver for the coset weight problem}
\label{CWGI1}

The coset weight problem is as follows. 

\begin{problem}[framed]{Coset Weight Problem (CWP)}
Instance: & $\mA\in\cM_{m,n}(\bF)$, $\vb\in\bF^m$,
            and positive integer $t$, where $\bF$ is a finite field; \\
Question: & Is there any solution $\vx_0\in\bF^n$ to $\mA\vx=\vb$ such 
            that $\Hw{\vx_0}\leq t$?  
\end{problem}

In coding theory, CWP occurs in the context of 
{\em syndrome decoding}, where $\mA$ is a full rank matrix of size 
$m\times n$ with $m<n$ and $\vb$ is a syndrome. 
For this reason, it is also called the
{\em syndrome decoding} problem (more details about it are provided 
in Section \ref{ISD0}). 

CWP is NP-complete when $\bF=\bZ_2$ \cite{BeET1978}. 
However, both highlighting easy instances and constructing probabilistic 
polynomial-time algorithms to solve this problem can be of major 
importance when the problem is used to design secure cryptographic 
primitives.

In this section, we will analyze CWP through the generalized 
inverses of the matrix $\mA$. We will present the results, as much as 
possible, for the case of a general finite field $\bF$. But, where 
necessary, we will restrict the analysis to $\bF=\bZ_2$.

The main strategy is the following. Given a CWP instance $(\mA,\vb,t)$, 
we will compute GIs $\mX$ of the matrix $\mA$ and 
check the solution's weight. In fact, once a GI $\mX$ 
is computed, we have two approaches we can follow:
\begin{enumerate}
\item Consider only solutions $\vx=\mX\vb$;
\item Consider solutions $\vx=\mX\vb+(\mId-\mX\mA)\vc$, where 
      $\vc\in\bF^n$. 
\end{enumerate}

However, we will show that any solution to the system 
$\mA\vx=\vb$ can be expressed in the form $\mX\vb$, where 
$\mX\in\GI(\mA)$. 

\begin{theorem}\label{CWGI1-T01}
Let $\mA\in\cM_{m,n}(\bF_q)$ with full row rank and $\vb\in \cR(\mA)$
with $\vb\not=0$.  
Then, 
\begin{equation}
\nonumber
  \{\vx\in\bF_q^n\;|\;\mA\vx=\vb\}=\{\mX\vb\;|\;\mX\in\GI(\mA)\}.
\end{equation}
\end{theorem}
\begin{proof}
The inclusion ``$\supseteq$'' follows simply from the fact that 
$\mX\vb$ is a solution to $\mA\vx=\vb$, for any $\mX\in\GI(\mA)$.
Showing that the two sets have the same number of elements will 
end the proof.

It is straightforward to verify that 
$\left|\{\vx\in\bF_q^n\;|\;\mA\vx=\vb\}\right|=q^{n-m}$, 
since $\mA$ has full rank.
 
Let 
$\mP\mA\mQ=\begin{pmatrix}
     \mV&\mId_{m}
 \end{pmatrix}$.  
be a transformation of $\mA$. 
By Corollary \ref{GI-C01}, the sets $\GI(\mA)$ and $\GI(\mP\mA\mQ)$ 
are isomorphic and hence we can restrict to evaluate 
$\left|\{\mX\vb\;|\;\mX\in\GI(\mP\mA\mQ)\}\right|$. 
Any GI of $\mP\mA\mQ$ has the form 
$\mX=\begin{pmatrix}
         \mX_1\\\mId_m-\mV\mX_1
     \end{pmatrix}$, 
where $\mX_1\in\cM_{n-m,m}(\bF_q)$ (Example \ref{GI-Eq01}(3)).
As $\vb\not=\mO$, $\mX_1\vb$ can take any value in $\bF_q^{n-m}$.
So, the number of solution 
$\mX\vb=\begin{pmatrix}
         \mX_1\vb\\ (\mId_m-\mV\mX_1)\vb
\end{pmatrix}$ 
is exactly $q^{n-m}$. This shows that the two sets have the same number 
of elements. 
\end{proof}

Hence, a generic GID solver for CWP samples 
$\mX\gets\GI(\mA)$ until $\Hw{\mX\vb}\leq t$. 
The main problem we get is how to do the sampling. 
Our approach is to apply transformations to $\mA$ to easily calculate 
GIs of the transformed matrix and to transfer the
result to $\mA$. 

Theorem \ref{GI-T03} tells us that a single transformation of  
$\mA$ suffices to generate all its GIs. 
Our simulations and previous results on existing ISD algorithms show 
that it is more efficient to use different transformations for 
$\mA$ and run through several GIs for each transformation. 
According to this, we present below a generic GID solver for CWP.

\begin{algorithm}[!h]
    \caption{GID solver for CWP}
    \label{CWGI1-Alg01} 
    \begin{algorithmic}[1]
        \Function{CWGI\_solve}{$\mA,\vb,t$} 
        \Repeat
        \State Choose a transformation $(\mP,\mQ)$ of $\mA$;
        \State $\mX\gets\GI_{\mP,\mQ}(\mA)$ until $\Hw{\mX\vb}\leq t$
               or no more sampling is allowed;
        \Until a solution $\mX\vb$ is found or no more transformation is
               allowed;
        \State\textbf{return} solution $\mX\vb$ or ``fail''. 
        \EndFunction
    \end{algorithmic}
\end{algorithm}

It should be understood that steps 3 and 4 are performed under 
various strategies, each leading to a variant of this generic algorithm.
This will be clear in Section \ref{ISD} when we discuss Prange, 
Lee-Brickell, Leon, Stern, and Finiasz-Sendrier's algorithms.

In the following, we will analyze some transformations that can 
be applied to the matrix $\mA$, focusing on the following three 
aspects:
\begin{itemize}
\item The general form of a GI $\mX$ of $\mA$;
\item Particularities of the solution $\mX\vb$;
\item Algorithmic issues.
\end{itemize}

\subsubsection{A general rank case}\label{CWGI1.1}

Our first result here, regarding the form of the GI,
follows directly from Corollary \ref{GI-C01} and Proposition \ref{GI-P01}.

\begin{corollary}\label{CWGI-C01}
Let $\mA\in\cM_{m,n}(\bF)$, $\mP\in\cGL_m(\bF)$, $\mQ\in\cGL_n(\bF)$,
and $r>0$ be an integer such that $r\leq \rank(\mA)\leq min\{m,n\}$ and 
\begin{equation}\label{CWGI-Eq01}
\mP\mA\mQ=\begin{pmatrix}
        \mId_r & \mA_2 \\
        \mO   & \mA_4
    \end{pmatrix}.
\end{equation}
Then, any GI $\mX$ of $\mA$ is of the form 
\begin{equation}\label{CWGI-Eq02}
\mX=\mQ \begin{pmatrix}
      \mX_1 & \mX_2 \\
      \mX_3 & \mX_4
    \end{pmatrix}
    \mP, 
\end{equation}
where $\mX_1$, $\mX_2$, $\mX_3$, and $\mX_4$ are matrices of appropriate 
sizes that verify the matrix equations:
\begin{equation}\label{CWGI-Eq03}
\left\{
\begin{aligned}
\mX_1 + \mA_2\mX_3 &= \mId_r \\
\mA_4\mX_3 &= \mO \\
(\mX_2 + \mA_2\mX_4)\mA_4 &= \mO \\ 
\mA_4\mX_4\mA_4 &= \mA_4
\end{aligned}
\right.
\end{equation} 
\end{corollary}

\begin{remark}\label{CWGI-R01}
The transformation \eqref{CWGI-Eq01} can always be obtained by 
Gaussian elimination. If no particular constraints are imposed on 
the matrices $\mA_2$ and $\mA_4$, $\mQ$ can be obtained as a 
permutation matrix.
It is also interesting to remark that $\mX_4$ in \eqref{CWGI-Eq02}
must be a GI of $\mA_4$ (see \eqref{CWGI-Eq03}). 
\end{remark}

\begin{lemma}\label{CWGI-L01}
With the above notation, if the system $\mA\vx=\vb$ is consistent, then 
\begin{equation}\label{CWGI-Eq04}
\mX\vb = \mQ \begin{pmatrix} 
         \mX_1\bar{\vb}' + \mX_2\bar{\vb}'' \\
         \mX_3\bar{\vb}' + \mX_4\bar{\vb}''
       \end{pmatrix},
\end{equation}
where $\bar{\vb}'=(\mP\vb)_{[1,r]}$ and $\bar{\vb}''=(\mP\vb)_{[r+1,m]}$.
\end{lemma}
\begin{proof}
Directly from (\ref{CWGI-Eq02}). 
\end{proof}

Equation (\ref{CWGI-Eq04}) gives us some flexibility in choosing the 
matrix $\mX$'s blocks to minimize the solution $\mX\vb$'s weight. 
The following procedure could be used:
\begin{itemize}
    \item Sample $\mX_4$ from $\GI(\mA_4)$;
    \item Compute $\mX_2$ by $\mX_2=-\mA_2\mX_4$;
    \item Generate a matrix $\mX_3$ such that $\mA_4\mX_3=\mO$;
    \item Compute $\mX_1$ by $\mX_1=\mId_r-\mA_2\mX_3$.
\end{itemize}
The procedure is repeated as long as $\Hw{\mX\vb}$ is greater than 
some given threshold $t$. 
The weight of $\mX\vb$ depends on $\mP$, $\mQ$, and the  
matrix $\mX$'s blocks. If $\mQ$ is a permutation, then it can be 
neglected in choosing the matrix $\mX$'s blocks because it does 
not change the weight.

\begin{remark}\label{CWGI-R02}
The above analysis is kept, with minor modifications, also for the 
case where $\mId_r$ occupies another position in \eqref{CWGI-Eq01}. 
For example, if
\begin{equation}\label{CWGI-Eq05}
\mP\mA\mQ=\begin{pmatrix}
        \mA_1 & \mId_r \\
        \mA_3 & \mO
    \end{pmatrix}
\end{equation}
then the blocks of the GI $\mX$ must verify the 
matrix equations:
\begin{equation}\label{GI-Eq06}
\left\{
\begin{aligned}
\mA_1\mX_1 + \mX_3 &= \mId_r \\
\mA_3\mX_1 &= \mO \\
(\mA_1\mX_2 + \mX_4)\mA_3 &= \mO \\ 
\mA_3\mX_2\mA_3 &= \mA_3
\end{aligned}
\right.
\end{equation}

In this case, $\mX_2$ is a GI of $\mA_3$, and 
Lemma \ref{CWGI-L01} holds true. 
\end{remark}

\subsubsection{Rank deficient matrices}\label{CWGI1.2}

This is a sub-case of the previous case, which deals with rank 
deficient matrices. Our first result follows directly from
Corollary \ref{GI-C01} and Proposition \ref{GI-P01}.

\begin{corollary}\label{CWGI-C02}
Let $\mA\in\cM_{m,n}(\bF)$, $\mP\in\cGL_m(\bF)$, and $\mQ\in\cGL_n(\bF)$
such that $rank(\mA)=r<min\{m,n\}$ and  
\begin{equation}\label{CWGI-Eq07}
\mP\mA\mQ=\begin{pmatrix}
        \mId_r & \mA_2 \\
        \mO   & \mO
    \end{pmatrix}.
\end{equation}
Then, any GI $\mX$ of $\mA$ is of the form 
\begin{equation}\label{CWGI-Eq08}
\mX=\mQ \begin{pmatrix}
      \mX_1 & \mX_2 \\
      \mX_3 & \mX_4
    \end{pmatrix}
    \mP, 
\end{equation}
where $\mX_1$, $\mX_2$, $\mX_3$, and $\mX_4$ are matrices of appropriate 
sizes that verify the matrix equation
\begin{equation}\label{CWGI-Eq09}
\mX_1 + \mA_2\mX_3 = \mId_r. 
\end{equation}
\end{corollary}

\begin{remark}\label{CWGI-R03}
The transformation \eqref{CWGI-Eq07} can always be obtained by 
Gaussian elimination. If no requirement is imposed on the matrix 
$\mA_2$, $\mQ$ can be obtained as a permutation matrix.
\end{remark}

\begin{lemma}\label{CWGI-L02}
With the above notation, if the system $\mA\vx=\vb$ is consistent, then 
\begin{equation}\label{CWGI-Eq10}
\mX\vb = \mQ \begin{pmatrix} 
         \mX_1\bar{\vb}' \\
         \mX_3\bar{\vb}'
       \end{pmatrix},
\end{equation}
where $\bar{\vb}'=(\mP\vb)_{[1,r]}$. 
\end{lemma}
\begin{proof}
According to Lemma \ref{CWGI-L01}, we only have to prove that 
$\bar{\vb}''=\mO$, where $\bar{\vb}''=(\mP\vb)_{[r+1,m]}$.

Let $\vx_0$ be an arbitrary but fixed solution to $\mA\vx=\vb$. That is,
$\mA\vx_0=\vb$. Then,
$$
\begin{array}{lcl}
  \mP\vb & = & \mP\mA\vx_0 \\ [.5ex]
     & = &\mP\mP^{-1}\begin{pmatrix}
                     \mId_r & \mA_2 \\     
                     \mO   & \mO 
                  \end{pmatrix}
                  \mQ^{-1}\vx_0 \\ [2ex]
     & = & \begin{pmatrix}
             \tilde{\vx}_0'+\mA_2\tilde{\vx}_0'' \\
             \mO
           \end{pmatrix}, 
\end{array}
$$
where $\tilde{\vx}_0'=(\mQ^{-1}\vx_0)_{[1,r]}$ and 
$\tilde{\vx}_0''=(\mQ^{-1}\vx_0)_{[r+1,n]}$. Therefore, $\bar{\vb}''=\mO$. 
\end{proof}

When $\mQ$ is a permutation, Lemma \ref{CWGI-L02} allows us to obtain 
solutions to $\mA\vx=\vb$ with any desired weight distribution
on $\pi_{\mQ}([r+1,n])$, as the next theorem shows. 

\begin{theorem}\label{CWGI-T01}
Let $\mA\in\cM_{m,n}(\bF)$ with $\rank(\mA)=r<\min\{m,n\}$, $\vb\in\bF^m$, 
$\mP\in\cGL_{m}(\bF)$, and $\mQ\in\cS_{n}(\bF)$ such that 
$(\mP\vb)_{[1,r]}\not=\mO$ and 
\begin{equation}
    \mP\mA\mQ=\begin{pmatrix}
             \mId_r & \mA_2 \\
             \mO   & \mO 
           \end{pmatrix}
\end{equation}
Then, for any set $\cI\subseteq\pi_{\mQ}([r+1,n])$, a solution $\vx$ to 
$\mA\vx=\vb$ with the property
$\supp{\vx}|_{\pi_{\mQ}([r+1,n])}=\cI$, can efficiently be computed. 
\end{theorem}
\begin{proof}
Let $\mA$, $\vb$, $\mP$, and $\mQ$ as in the theorem's hypothesis. 
Let $\mX$ be a GI of $\mA$. 
Then, Lemma \ref{CWGI-L02} leads to 
$$
    \mX\vb = \mQ \begin{pmatrix}
                   \mX_1 (\mP\vb)_{[1,r]} \\
                   \mX_3 (\mP\vb)_{[1,r]}
                 \end{pmatrix}, 
$$
where $\mX_1$ and $\mX_3$ verify the matrix equation
$\mX_1 + \mA_2\mX_3 = \mId_r$. 

Now, let $\cI$ be a subset of $\pi_{\mQ}([r+1,n])$, $i\in\cI$, and 
$1\leq j\leq r$ be such that $(\mP\vb)(j)\not=0$.  
Then, there exists $r+1\leq k\leq n$ such that $\pi_{\mQ}(k)=i$. 
Moreover, 
  $$(\mX\vb)(i)=\sum_{j=1}^r {\mX_3}(k-r,j)(\mP\vb)(j).$$
We choose now the block $\mX_3$ by
$$
{\mX_3}(k-r,\ell)= 
\begin{cases}
 1, & \mbox{if $\pi_{\mQ}(k)=i\in\cI$ and $\ell=j$} \\
 0, & \mbox{otherwise,}
\end{cases}
$$
for all $r+1\leq k\leq n$ and $1\leq\ell\leq r$. 
Then, $\mX_1=\mId_r-\mA_2\mX_3$. 

It is straightforward to check that $\vx=\mX\vb$ verifies
$\supp{\vx}|_{\pi_{\mQ}([r+1,n])}=\cI$. 
Moreover, $\vx$ can be computed in polynomial time in the size of $\mA$. 
\end{proof}

\begin{remark}\label{CWGI-R04}
Under the conditions of Theorem \ref{CWGI-T01}, assuming in addition 
$\mA_2=\mO$, the solution $\mX\vb$ will be of form
\begin{equation}
\mX\vb=\mQ\begin{pmatrix} 
             (\mP\vb)_{[1,r]} \\ 
             \mX_3(\mP\vb)_{[1,r]}
    \end{pmatrix}
\end{equation}
because $\mX_1=\mId_r$. 
Therefore, for any $\cI$ chosen as in Theorem \ref{CWGI-T01} we can
compute a solution $\vx$ such that 
$\supp{\vx}=\cI\cup \pi_{\mQ}(\supp{(\mP\vb)_{[1,r]}}).$ 
\end{remark}

\begin{remark}\label{CWGI-R05}
When $\mP\mA\mQ$ takes the form
\begin{equation}\label{CWGI-Eq11}
\mP\mA\mQ=\begin{pmatrix}
        \mA_1 & \mId_r \\
        \mO   & \mO
    \end{pmatrix}, 
\end{equation}
the blocks of the GI must verify the matrix equation 
\begin{equation}\label{CWGI-Eq12}
\mA_1\mX_1 + \mX_3 = \mId_r.
\end{equation}

The solution has the same form as in equation \eqref{CWGI-Eq10}. 
This time however, in Theorem \ref{CWGI-T01}, we will consider 
$\cI\subseteq\pi_{\mQ}([1,n-r])$ and in the proof we will assign 
$\mX_1$ instead of $\mX_3$.

A conclusion similar to that in Remark \ref{CWGI-R04} can be 
drawn when $\mA_1=0$.
\end{remark}

\subsubsection{Full rank matrices}\label{CWGI1.3}

This is another particular case of the one in Section \ref{CWGI1.1}. 
As in previous cases, our first result follows directly from
Corollary \ref{GI-C01} and Proposition \ref{GI-P01}.

\begin{corollary}\label{CWGI-C03}
Let $\mA\in\cM_{m,n}(\bF)$, $\mP\in\cGL_m(\bF)$, and 
$\mQ\in\cGL_n(\bF)$ such that $rank(\mA)=r=m<n$ and  
\begin{equation}\label{CWGI-Eq13}
\mP\mA\mQ=\begin{pmatrix}
        \mId_r & \mA_2 
    \end{pmatrix}.
\end{equation}
Then, any GI $\mX$ of $\mA$ is of the form 
\begin{equation}\label{CWGI-Eq14}
\mX=\mQ \begin{pmatrix}
      \mX_1  \\
      \mX_2 
    \end{pmatrix}
    \mP, 
\end{equation}
where $\mX_1$ and $\mX_2$ are matrices of appropriate 
sizes that verify the equation
\begin{equation}\label{CWGI-Eq15}
\mX_1 + \mA_2\mX_2 = \mId_r.
\end{equation}
\end{corollary}

\begin{remark}\label{CWGI-R06}
\begin{enumerate}
\item The transformation \eqref{CWGI-Eq13} can always be obtained by 
  Gaussian elimination. If no requirement is imposed on the matrix 
  $\mA_2$, $\mQ$ can be obtained as a permutation matrix.
\item It is straightforward to see that Lemma \ref{CWGI-L02} and 
  Theorem \ref{CWGI-T01} remain valid in this case as well.
  Please also remark that $\bar{\vb}'$ from Lemma \ref{CWGI-L02}
  is now exactly $\bar{\vb}$ because $r=m$. 
\end{enumerate}
\end{remark}

\begin{remark}\label{CWGI-R07}
If $\mA_2=0$ in Corollary \ref{CWGI-C03}, then the only constraint 
the GI must satisfy is $\mX_1=\mId_r$. 
The solution $\mX\vb$ is then
\begin{equation}
\mX\vb=Q\begin{pmatrix} 
             \bar{\vb} \\ \mX_2\bar{\vb}
    \end{pmatrix}
\end{equation}
where for the choice of $\mX_3$ we have full flexibility.
The conclusion in Theorem \ref{CWGI-T01} can then be strengthen to
``$Supp(\vx)=\cI\cup \pi_Q(Supp((\mP\vb)[1,r]))$'', as we did in 
Remark \ref{CWGI-R04}. 
\end{remark}

\begin{remark}\label{CWGI-R08}
If we take
\begin{equation}\label{CWGI-Eq16}
\mP\mA\mQ=\begin{pmatrix}
        \mA_1 & \mId_r 
    \end{pmatrix}
\end{equation}
then the blocks of the GI $\mX$ should satisfy
\begin{equation}\label{CWGI-Eq17}
\mA_1\mX_1 + \mX_2 = \mId_r. 
\end{equation}
The discussion on the solution $\mX\vb$ is then as in 
Remark \ref{CWGI-R04} with $\mP\vb$ instead of $(\mP\vb)_{[1,r]}$. 

A conclusion similar to that in Remark \ref{CWGI-R07} can be 
drawn when $\mA_1=0$.
\end{remark}

\begin{remark}
A similar discssion to that above stands for the case when
$rank(\mA)=r=n<m$. If we assume that 
\begin{equation}\label{CWGI-Eq18}
\mP\mA\mQ=\begin{pmatrix}
        \mId_r \\ \mA_2 
    \end{pmatrix}
\end{equation}
then any GI $\mX$ of $\mA$ is of the form 
\begin{equation}\label{CWGI-Eq19}
\mX=\mQ \begin{pmatrix}
      \mX_1 & \mX_2 
    \end{pmatrix}
    \mP, 
\end{equation}
where $\mX_1$ and $\mX_2$ verify the equation
\begin{equation}\label{CWGI-Eq20}
\mX_1 + \mX_2\mA_2 = \mId_r. 
\end{equation}

The solution $\mX\vb$ has the form
$\mX\vb=\mQ(\mX_1(\mP\vb)_{[1,r]}+\mX_2(\mP\vb)_{[r+1,m]})$ 
(please also see Lemma \ref{CWGI-L01}). 

If we take
\begin{equation}\label{CWGI-Eq21}
\mP\mA\mQ=\begin{pmatrix}
        \mA_1 \\ \mId_r 
    \end{pmatrix}
\end{equation}
then the blocks of the GI $\mX$ should satisfy
\begin{equation}\label{CWGI-Eq22}
  \mA_1\mX_1 + \mX_2 = \mId_r. 
\end{equation}
\end{remark}

\subsubsection{Summarizing the results}\label{CWGI1.4}

The table in Figure 1 summarizes the form of solutions for the case 
of rank deficient and full rank matrices. The fourth column of the 
table also includes information about the matrices' sizes to help the 
reader get a pictorial view of them. 

\begin{figure}[hbt]
\begin{center}
\begin{tabular}{|C{1cm}|C{3cm}|C{4cm}|C{3cm}|}
\toprule
    & Transformation $\mP\mA\mQ$ & GI $\mX$ & Solution $\mX\vb$ \\ 
\midrule 
 1.  & $\begin{pmatrix} \mId_r & \mA_2 \\ \mO & \mO \end{pmatrix}$
     & \makecell{
          $\mQ\begin{pmatrix} \mX_1 & \mX_2 \\ 
                              \mX_3 & \mX_4 \end{pmatrix}\mP$ \\ 
          $\mX_1+\mA_2\mX_3=\mId_r$
                 }
     & $\mQ\begin{pNiceMatrix}[first-row,last-col] 
              {\color{red} r} & \\ 
              \mX_1\bar{\vb}' & {\color{red} r} \\ 
              \mX_3\bar{\vb}' & {\color{red} n-r} 
        \end{pNiceMatrix}$ \\
\midrule 
 2.  & $\begin{pmatrix} \mId_r & \mO \\ \mO & \mO \end{pmatrix}$
     & $\mQ\begin{pmatrix} \mId_r & \mX_2 \\ 
                           \mX_3 & \mX_4 \end{pmatrix}\mP$ 
     & $\mQ\begin{pNiceMatrix}[first-row,last-col] 
              {\color{red} r} & \\ 
              \bar{\vb}' & {\color{red} r} \\ 
              \mX_3\bar{\vb}' & {\color{red} n-r} 
           \end{pNiceMatrix}$ \\
\midrule 
 3.  & $\begin{pmatrix} \mId_r & \mA_2 \end{pmatrix}$
     & \makecell{ 
           $\mQ\begin{pmatrix} \mX_1 \\ \mX_2 \end{pmatrix}\mP$ \\ 
           $\mX_1+\mA_2\mX_2=\mId_r$
                 } 
     & $\mQ\begin{pNiceMatrix}[first-row,last-col]
               {\color{red} r} & \\
               \mX_1\bar{\vb}  & {\color{red} r} \\ 
               \mX_2\bar{\vb}  & {\color{red} n-r} 
           \end{pNiceMatrix}$ \\
\midrule 
 4.  & $\begin{pmatrix} \mId_r & \mO \end{pmatrix}$
     & $\mQ\begin{pmatrix} \mId_r \\ \mX_2 \end{pmatrix}\mP$ 
     & $\mQ\begin{pNiceMatrix}[first-row,last-col]
               {\color{red} r} & \\ 
               \bar{\vb}       & {\color{red} r} \\ 
               \mX_2\bar{\vb}  & {\color{red} n-r} 
           \end{pNiceMatrix}$ \\
\midrule 
 5.  & $\begin{pmatrix} \mA_1 & \mId_r \\ \mO & \mO \end{pmatrix}$
     & \makecell{ 
           $\mQ\begin{pmatrix} \mX_1 & \mX_2 \\ \mX_3 & \mX_4 \end{pmatrix}\mP$ \\
           $\mA_1\mX_1+\mX_3=\mId_r$ 
                 }
     & $\mQ\begin{pNiceMatrix}[first-row,last-col] 
               {\color{red} r} & \\
               \mX_1\bar{\vb}' & {\color{red} n-r} \\ 
               \mX_3\bar{\vb}' & {\color{red} r}
           \end{pNiceMatrix}$ \\
\midrule 
 6.  & $\begin{pmatrix} \mO & \mId_r \\ \mO & \mO \end{pmatrix}$
     & $\mQ\begin{pmatrix} \mX_1 & \mX_3 \\ \mId_r & \mX_4 \end{pmatrix}\mP$ 
     & $\mQ\begin{pNiceMatrix}[first-row,last-col] 
               {\color{red} r} & \\
               \mX_1\bar{\vb}' & {\color{red} n-r} \\ 
               \bar{\vb}'      & {\color{red} r}
           \end{pNiceMatrix}$ \\
\midrule 
 7.  & $\begin{pmatrix} \mA_1 & \mId_r \end{pmatrix}$
     & \makecell{
           $\mQ\begin{pmatrix} \mX_1 \\ \mX_2 \end{pmatrix}\mP$ \\ 
           $\mA_1\mX_1+\mX_2=\mId_r$ 
                 }
     & $\mQ\begin{pNiceMatrix}[first-row,last-col]
               {\color{red} r} & \\
               \mX_1\bar{\vb}  & {\color{red} n-r} \\ 
               \mX_2\bar{\vb}  & {\color{red} r}
      \end{pNiceMatrix}$ \\
\midrule 
 8.  & $\begin{pmatrix} \mO & \mId_r \end{pmatrix}$
     & $\mQ\begin{pmatrix} \mX_1 \\ \mId_r \end{pmatrix}\mP$ 
     & $\mQ\begin{pNiceMatrix}[first-row,last-col]
               {\color{red} r} & \\
               \mX_1\bar{\vb}  & {\color{red} n-r} \\ 
               \bar{\vb}       & {\color{red} r}
            \end{pNiceMatrix}$ \\
\bottomrule 
\end{tabular}
\end{center}
\caption{$\mA\in\cM_{m,n}(\bF)$, $\vb\in\bF^m$, 
         $\bar{\vb}=\mP\vb$, $\bar{\vb}'=(\mP\vb)_{[1,r]}$} 
\label{CWGI1.4-F01}
\end{figure}

Theorem \ref{CWGI-T05} (and the cases deduced from it) is vital in 
enumerating the solutions defined by the GI. We want 
to emphasize this here. Suppose that the GI is as
in Figure \ref{CWGI1.4-F01}, row 7. Let $\bar{\vb}$ be a non-zero vector.
The matrix $\mX_1$ can be chosen so that $\mX_1\bar{\vb}$ has any distribution 
of its non-zero components. Therefore, to obtain a specific distribution 
or weight of $\mX_2$, we will have to look for a vector $\vz_1$ so 
that $\bar{\vb}-\mA_1\vz_1$ is what we want. We can then easily compute 
$\mX_1$ to satisfy the relation $\mX_1\bar{\vb} = \vz_1$.

\subsection{The subspace weight problem and a generic GID solver for it}
\label{CWGI2}

An important sub-problem of CWP is the one obtained by 
considering $\vb=\v0$.

\begin{problem}[framed]{Subspace Weight Problem (SWP)}
Instance: & $\mA\in\cM_{m,n}(\bF)$ and positive integer $t$, 
            where $\bF$ is a finite field; \\
Question: & Is there any solution $\vx_0\in\bF^n$ to $\mA\vx=\v0$ such 
            that $\Hw{\vx_0}\leq t$?  
\end{problem}

In coding theory, SWP is usually related to the low weight codeword problem
that asks to find a codeword of small weight, of the code whose 
parity-check matrix is $\mA$ (more details about it are provided 
later in Section \ref{ISD0}). 

Even if SWP is a hard sub-problem of CWP, which is 
NP-complete, it is not immediately apparent that it is also NP-complete.
It was conjectured in \cite{BeET1978} that it is NP-complete,
but the proof was later provided in \cite{Vard1997}.  

Using GIs to solve SWP only gives us 
the solution $\mX\v0=\v0$. Characterizing the kernel of the matrix 
$\mA$ by $\cN(\mA) = \cR(\mId-\mX\mA)$, where $\mX$ is a GI 
of $\mA$, could be a solution. It requires sampling vectors 
$\vc$ and computing values $(\mId-\mX\mA)\vc$. 

We will present below a method that we believe offers a good 
potential to approach the problem. 

\begin{theorem}\label{CWGI2-T01} 
Let $\mA\in\cM_{m,n}(\bF_q)$ with full row rank, $\vb\in \cR(\mA)$
with $\vb\not=\mO$, and $\mX\in\GI(\mA)$. Then,
\begin{equation}
  \cN(\mA)=\{(\mY-\mX)\vb\mid \mY\in\GI(\mA)\}.
\end{equation}
\end{theorem}
\begin{proof}
It is straightforward to see that 
$\{(\mY-\mX)\vb\mid \mY\in\GI(\mA)\}\subseteq\cN(\mA)$.

Conversely, if $\vv\in\cN(\mA)$, then $\vx=\mX\vb+\vv$ is a
solution to $\mA\vx=\vb$ (Theorem \ref{GI-T02}). On the other
side, there exists a GI $\mY$ of $\mA$ such
that $\vx=\mY\vb$ (Theorem \ref{CWGI1-T01}). Combining the two forms
of $\vx$ we get $\vv=(\mY-\mX)\vb$. 
\end{proof}

The proof of Theorem \ref{CWGI2-T01} uses Theorem \ref{CWGI1-T01}.
One can prove Theorem \ref{CWGI2-T01} first, based on a counting argument 
as in the proof of Theorem \ref{CWGI1-T01}, and then derive 
Theorem \ref{CWGI1-T01}. 

The shape of $\mA$'s kernel elements in Theorem \ref{CWGI2-T01} 
depends on the shape of the GIs. Therefore, it
depends on the transformation applied to the matrix $\mA$
(please see our discussion on this topic at the end of Section \ref{GI2}). 
Corollary \ref{CWGI2-C01} below illustrates our discussion.

\begin{corollary}\label{CWGI2-C01}
Let $\mA\in\cM_{m,n}(\bF_q)$ with full row rank, $\vb\in \cR(\mA)$
with $\vb\not=0$, and 
$\mP\mA\mQ=\begin{pmatrix} \mV & \mId_m \end{pmatrix}$ be a
transformation of $\mA$. Then, 
\begin{equation}
  \cN_{\mP,\mQ,\vb}(\mA)=
        \left\{\mQ\begin{pmatrix} \mZ \\ -\mV\mZ \end{pmatrix}\mP\vb
               \bigm| \mZ\in\cM_{n-m,m}(\bF_q)
        \right\}
\end{equation}
\end{corollary}
\begin{proof}
Let $\mX=\mQ\begin{pmatrix} \mX_1 \\ \mId_m-\mV\mX_1 \end{pmatrix}\mP$
be an arbitrary but fixed GI of $\mA$. 
Then, according to Theorem \ref{CWGI2-T01},
  $$\cN_{\mP,\mQ,\vb}(\mA)=\{(\mY-\mX)\vb\mid \mY\in\GI(\mA)\}.$$
For any $\mY\in\GI(\mA)$ there exits $\mY_1$ such that 
$\mY=\mQ\begin{pmatrix} \mY_1 \\ \mId_m-\mV\mY_1 \end{pmatrix}\mP$.
Then,
\begin{equation}
 (\mY-\mX)\vb=\mQ\begin{pmatrix} \mY_1-\mX_1 \\ 
                -\mV(\mY_1-\mX_1) \end{pmatrix}\mP\vb.
\end{equation}
When $\mY_1$ spans $\cM_{n-m,m}(\bF_q)$, $\mZ=\mY_1-\mX_1$ spans
the same space because $\mX_1$ is fixed. So, we get the result.
\end{proof}

Theorem \ref{CWGI2-T01} allows us to design algorithms for SWP 
similar to those for CWP. 
Algorithm \ref{CWGI2-Alg02} below is a generic GID solver for SWP.

\begin{algorithm}[!h]
    \caption{GID solver for SWP}
    \label{CWGI2-Alg02}
    \begin{algorithmic}[1]
        \Function{SWGI\_solve}{$\mA,t$} 
        \Repeat
        \State Choose a transformation $(\mP,\mQ)$ of $\mA$;
        \State Choose $\vb\in\cR(\mA)$ with $\vb\not=0$;
        \State $\vv\gets\cN_{\mP,\mQ,\vb}(\mA)$ until $\Hw{\vv}\leq t$
               or no more sampling is allowed;
        \Until a solution $\vv$ is found or no more transformation is
               allowed;
        \State\textbf{return} solution $\vv$ or ``fail''. 
        \EndFunction
    \end{algorithmic}
\end{algorithm}

It should be understood that steps 3 and 5 in Algorithm \ref{CWGI2-Alg02}
are performed under 
various strategies, each leading to a variant of this generic algorithm.
The vector $\vb$ can be obtained by making the product between $\mA$ 
and a randomly chosen vector $\vx_0$. $\vb$ must be non-zero and 
preferably have as little weight as possible to reduce computational 
costs.

Algorithm \ref{CWGI2-Alg02} will be illustrated on Lee-Brickell, Leon, 
and Stern's approaches in Section \ref{ISD}.

\section{Information set decoding and GID decoding}\label{ISD}

In this section, we will show that many of the information set
decoding (ISD) techniques used to solve the syndrome decoding
problem in the theory of linear codes fit very well into the 
GID approach.
In addition to the fact that the use of the GI 
provides a clearer overview of the problem, it also allows for 
improvements to the ISD techniques proposed so far.

\subsection{Information set decoding as a cryptographic attack}\label{ISD0}

``Information set decoding'' refers to a class of techniques used in 
coding theory, especially in decoding problems, based on the concept 
of an information set.
To better understand things, we will make a brief historical foray 
into developing the concept of information set decoding and its 
connection with cryptography.

The concept of information set was introduced in 1962 by Prange 
\cite{Pran1962} in its attempt to propose a method of decoding 
cyclic codes. The basic idea was that, given a word $w$ to be decoded, 
to build a set of code words that match $w$ on a certain set of 
positions called {\em information set}. 
Formally, the concept of information set is as follows. 
A {\em systematic code} is an [n, k] code in which every code word 
can be separated into $k$ information symbols and $(n–k)$ check 
symbols. The information symbols are identical with those of the 
source message before encoding. Thus, the process of encoding a 
systematic code involves the insertion of $(n–k)$ check symbols 
among the information symbols. The insertion positions must be 
the same for all the code words in the code. The set of positions 
of the information symbols is called an {\em information set}. 
Every linear code can be arranged to be systematic.

Generally speaking, the techniques based on the concept of information 
set locate a potential information-set in a generator matrix, 
parity-check matrix, codeword, and perform various processing on it. 
Next, we will see the connection between these techniques and 
cryptanalysis.

Public-key cryptography was born in May 1975 (according to  
\cite{Diff1988}) as the ``child of two problems'',  key distribution 
and digital signature.
Diffie and Helmann's 1976 paper \cite{DiHe1976} has hugely 
impacted cryptography and information security by proposing 
the key exchange method known today as the 
{\em Diffie-Hellman key exchange} and suggesting the 
public-key cryptosystem and digital signature concepts. 
One immediately noted the importance and necessity of one-way 
functions to design public-key cryptographic primitives. 
It is not surprising that the 1976-1978 period was characterized 
by an effervescent proposal of public-key cryptographic 
primitives.
Thus, in March 1977 (published May 1978), Berlekamp et al. showed 
that the problem of general decoding linear codes is NP-complete. 
In Jan-Feb 1978, McEliece proposes a public-key cryptosystem 
that bases its security on the problem of general decoding 
\cite{McEl1978}. 
For a long time, this cryptosystem was seen only as an 
``alternative to public-key cryptography based on number theory''. 
Things changed significantly when it was understood that this 
cryptosystem could provide security against quantum attacks, 
which is not the case with the systems based on factorization 
or discrete logarithm. This fact has led to McEliece's 
cryptosystem receiving much attention lately, being 
re-evaluated, re-analyzed, and looking for ways to make 
it more efficient.

\begin{cipher}[McEliece's cryptosystem \cite{McEl1978}]
Let $\mG\in\cM_{k,n}(\bF_2)$
be a generator matrix for a $t$-error correcting $[n,k]$ linear code. 
Choose $\mS\in\cGL_{k}(\bF_2)$ and $\mP\in\cS_{n}(\bF_2)$ 
and compute $\bar{\mG}=\mS\mG\mP$.
$\bar{\mG}$ is the public key. To encrypt a message
$\vm\in\bF_2^k$ randomly generate an $n$-bit vector $\ve$ with 
$\Hw{\ve}\leq t$ and compute 
\begin{equation}\label{ISD0-Eq01}
  \vc^t=\vm^t\bar{\mG}+\ve^t.
\end{equation}
To decrypt $\vc$, compute $\vc^t\mP^{-1}$ and then use a decoding
algorithm for $\mG$ to get rid of $\ve^t\mP^{-1}$ (remark that 
$\ve^t\mP^{-1}$ has Hamming weight at most $t$) and recover 
$\vm^t\mS$. Finally, get $\vm$ by means of $\mS^{-1}$. 
\end{cipher}

McEliece's security analysis of the proposed cryptosystem links 
to information set decoding, even if not explicitly. Thus, he 
says in \cite{McEl1978}: 
``A more promising attack is to select $k$ of the $n$ coordinates
randomly in hope that none of the $k$ are in error, and based on 
this assumption, to calculate $u$''. 
The method was not much exploited by McEliece but is re-discussed 
in \cite{RaNa1986} and \cite{AdMe1987}.

Let $\mA_I$ denote the restriction of the matrix $\mA$ to a given
set $I$ of column positions, that is the matrix obtained from $\mA$ 
by removing all columns whose index is not in $I$. Similarly define 
$\vv^t_I$ for vectors $\vv$. 
Assume now that $I$ has cardinality $k$. Then, one can easily obtain
  $$\vc^t_I=\vm^t\bar{\mG}_I+\ve^t_I.$$
If $I$ is a set of error free positions (that is, $I$ is an information
set), then $\ve^t_I=\v0^t$ and so,
one can recover $\vm$ by using $(\bar{\mG}_I)^{-1}$. 
As a result, the attack consists of repeating the following procedure:
\begin{itemize}
\item Randomly select $k$ positions in the ciphertext and restrict 
  it to them;
\item Apply the inverse of the matrix $\bar{\mG}$.
\end{itemize} 
If these $k$ positions are error-free (that is, they form an information
set), then the message $\vm$ is obtained. 
But how can we know that the resulting message is $\vm$? 
We will not go into details on this question, but we refer the reader
to \cite{LeBr1988} for an answer.

Let us consider now the parity-check matrix 
$\bar{\mH}\in\cM_{n-k,n}(\bF_2)$ associated to $\bar{\mG}$. 
If we multiply \eqref{ISD0-Eq01} to the right by $\bar{\mH}^t$,
we obtain
\begin{equation}
  \ve^t\bar{\mH}^t=\vs^t,
\end{equation}
where $\vs^t=\vc^t\bar{\mH}^t$ is the {\em syndrome} of $\vc^t$ and 
$\ve$ is the unknown vector. 

The above attack now focuses on determining the error instead of 
determining the message. 
But this is a specific attack on Niederreiter's cryptosystem \cite{Nied1986}.

\begin{cipher}[Niederreiter's cryptosystem \cite{Nied1986}]
Let $\mH\in\cM_{n-k,n}(\bF_2)$
be a parity-check matrix for a $t$-error correcting $[n,k]$ linear code. 
Choose $\mS\in\cGL_{k}(\bF_2)$ and $\mP\in\cS_{n}(\bF_2)$ and compute $\bar{\mH}=\mS\mH\mP$.
$\bar{\mH}$ is the public key. 
To encrypt $\ve\in \bF_2^n$ of weight $t$, compute 
\begin{equation}\label{ISD0-Eq02}
  \vc=\bar{\mH}\ve.
\end{equation}
To decrypt $\vc$, compute $\mS^{-1}\vc$ and then use a decoding
algorithm to recover $\mP\ve$. Finally, get $\ve$ by means of $\mP^{-1}$. 
\end{cipher}

So, what we have described above shows that if Niederreiter's cryptosystem  
can be easily broken, then McEliece's cryptosystem can be easily broken, 
and vice versa (for more details see \cite{LiDW1994}). 
As a result, solving matrix equation \eqref{ISD0-Eq02} becomes critical 
to the security of both cryptosystems.
Equation \eqref{ISD0-Eq02} is an instance of the 
SDP in the theory of linear codes. 

\begin{problem}[framed]{Syndrome Decoding Problem (SDP)}
Instance: & $\mH\in\cM_{n-k,n}(\bF_q)$ of full rank,
            syndrome $\vs\in\bF_q^{n-k}$, and positive integer $t$; \\
Question: & Is there a solution $\ve$ of Hamming weight 
            at most $t$ to the equation $\mH\ve=\vs$? 
\end{problem}

SDP is a particular case of CWP. Recall that, given an $[n,k]$ linear code 
$\cC$, a generator matrix $\mG$ and a
parity-check matrix $\mH$ of $\cC$, we may write 
\begin{equation}
  \cC = \{\vm^t\mG \mid \vm\in\bF_2^k\}
      = \{\vc^t\in\bF_2^n \mid \mH\vc=\v0\}
\end{equation}

We notice that the determination of a codeword of weight less than 
or equal to some positive integer $t$ is a sub-problem of SDP
(when the syndrome is zero) or, more appropriate, it is a 
particular case of SWP.  
This is often met as the {\em low weight codeword} problem or the 
{\em minimum weight codeword} problem, even if ``low weight'' is not 
necessarily ``minimum weight''.

\begin{problem}[framed]{Low Weight Codeword Problem (LWP)}
Instance: & Linear code $\cC$ specified by $\mG$ or $\mH$ and positive integer $t$; \\
Question: & Is there a nonzero codeword $c\in\cC$ such that 
            $\Hw{c}\leq t$? 
\end{problem}

LWP is NP-complete, while the computational version 
of it is NP-hard \cite{Vard1997}. 
This problem, both in decisional and computational form, is 
important in coding theory and cryptography. For instance, easily 
solving it leads to an attack on McEliece and Niederreiter's 
cryptosystems. 
Indeed, if $\vc$ is a McEliece ciphertext, then adding $\vc$ to 
the generator matrix $\bar{\mG}$ leads to a generator matrix 
$\tilde{\mG}$ for a $[k+1,n]$ linear code. 
In addition, there is a single codeword 
$d\in\langle\bar{\mG}\rangle\cap \langle\tilde{\mG}\rangle$ 
with Hamming distance $d(\vc,\vd)=t$. As a result, $\Hw{\vc-\vd}=t $. 
Moreover, no other codeword has weight $t$.
Once $\vc-\vd$ is computed, knowing $\vc$, we get $\vd$ which 
is undoubtedly the original message.

Over time, many algorithms for finding low-weight binary codewords 
have been proposed, based on the concept of information set 
(although this has not always been explicitly stated): 
\cite{Leon1988,Ster1988,CaCh1998,FiSe2009,BeLP2011,MaMT2011,BJMM2012,MaOz2015}. 
All these algorithms have been developed for the binary case. 
But there are also extensions to arbitrary finite fields, such as
\cite{CoGo1990,Pete2010,Meur2012,Hiro2016,NPCBB2017,WBSCBP2020,HoWe2021}. 

Throughout the next sections, we will denote by $(\mH,\vs,t)$ 
($(\mH,t)$) a generic instance of SDP (LWP), where 
$\mH\in\cM_{n-k,n}(\bF_q)$ is a full rank matrix, $\vs\in\bF_q^{n-k}$, 
and $t\leq n$ is a positive integer. We will also use the notation 
$r=n-k$, which is the rank of $\mH$.

\subsection{Prange's approach}\label{ISD1}

We will describe in this section Prange's approach \cite{Pran1962}
to SDP. 
We will then show that this is a particular case of our GID approach. 
We will show that Prange's ISD algorithm iterates through a smaller 
solution space than the total solution space. At the end of the section, 
we suggest improvements to Prange's method.

\subsubsection{Prange's algorithm}\label{ISD1.1}

The method proposed by Prange in the context of cyclic codes consists 
of the following. Having a word with possible errors, which we must 
decode, guess the error-free positions, and move them to the left, 
thus obtaining the error positions on the right. This method was also 
applied to McEliece's cryptosystem (see Section \ref{ISD0}) except 
that the cryptotext was restricted to error-free positions to decode it 
using the inverse of the public matrix (restricted to the same positions).

When using the parity-check matrix, Prange's method is equivalent to 
determining the error in an equation $\mH\ve=\vs$. In this case, packing 
the error-free positions to the left of $\mH$ means translating the error 
bits in $\ve$ to the base of the vector (leaving its top without error). 
This results immediately from the fact that the equation $\mH\ve = \vs$ 
is equivalent to $(\mH\mQ)(\mQ^{-1}\ve)=\vs$, for any permutation matrix 
$\mQ$.

Specifically, Prange's method in this context is as follows
($(\mH,\vs,t)$ is an instance of SDP):
\begin{enumerate}
\item Guess an information set $\cI\subseteq\{1,\ldots,n\}$ of size 
      $k$ and assume that each $i\in \cI$ is an error-free
      position;
\item Guess a permutation matrix $\mQ\in\cS_n(\bF)$ and compute a non-singular
      matrix $\mP\in\cGL_r(\bF)$ so that all columns indexed 
      by $\cI{}$ are packed to the left and $\mH$ is transformed into 
      \begin{equation}\label{ISD1-Eq01}
         \mP\mH\mQ=\begin{pmatrix} \mV & \mId_r \end{pmatrix}
      \end{equation}
      for some matrix $\mV$, where $r=n-k$ (if this is not possible, a different 
      permutation matrix is chosen). Remark that the information set $\cI$ 
      in $\mH$ becomes now the set of the first $k$ positions in $\mP\mH\mQ$
      ($\mV$ may not coincide with the sub-matrix of $\mH$ determined by $\cI$ 
      due to the use of $\mP$);  
\item The equivalences   
      $$
      \begin{array}{lcl}
       \mH\vx = \vs & \Lra & \mP^{-1} \begin{pmatrix} \mV & \mId_r \end{pmatrix} 
                           \mQ^{-1}\vx = \vs \\ [.5ex]
              & \Lra & \begin{pmatrix} \mV & \mId_r \end{pmatrix} \mQ^{-1}\vx = \mP\vs
      \end{array}
      $$
      show that $\vx$ is a solution to $\mH\vx=\vs$ if and only if $\vz=\mQ^{-1}\vx$
      is a solution to 
      \begin{equation}
      \label{ISD1-Eq02}
        \begin{pmatrix} \mV & \mId_r \end{pmatrix} \vz = \mP\vs. 
      \end{equation}
      
      However, (\ref{ISD1-Eq02}) becomes 
      \begin{equation}\label{ISD1-Eq03}
         \begin{pmatrix} \mV & \mId_r \end{pmatrix}
         \begin{pmatrix} \vz_1 \\ \vz_2 \end{pmatrix} 
         =\mP\vs
      \end{equation}
      and thus
      \begin{equation}\label{ISD1-Eq04}
        \mV\vz_1 + \vz_2 = \mP\vs, 
      \end{equation}
      where $\vz_1\in\bF^k$, $\vz_2\in\bF^r$, and 
      $\vz=\begin{pmatrix} \vz_1 \\ \vz_2 \end{pmatrix}$. 
      
      The assumption that $\vz_1$ is error-free because it corresponds to the 
      information set after permutation, leads to $\vz_1=\mO$ and $\vz_2=\mP\vs$. 
      Re-permuting $\vz$ we get 
      $\vx=\mQ\begin{pmatrix} \mO \\ \mP\vs \end{pmatrix}$, which is a 
      solution to $\mH\vx=\vs$. 
      Moreover, remark that $\Hw{\vx}=\Hw{\vz}=\Hw{\mP\vs}$ 
      because $\mQ$ is a permutation. 
\end{enumerate}

Let us refer to a pair $(\mP,\mQ)$ of matrices as in 
\eqref{ISD1-Eq01} as a {\em Prange transformation} of $\mH$. 
A solution obtained employing a Prange transformation $(\mP,\mQ)$
will be referred to as a {\em Prange $(\mP,\mQ)$-solution}. Clearly, 
it is uniquely determined by $(\mP,\mQ)$.
A {\em Prange solution} is a Prange $(\mP,\mQ)$-solution, for some 
Prange transformation $(\mP,\mQ)$.

\subsubsection{Prange's approach as a GI-based approach}\label{ISD1.2}

Prange's approach is a particular case of our generic GID solver
for CWP (Algorithm \ref{CWGI1-Alg01}). 
To see that, let us consider the transformation \eqref{ISD1-Eq01}.
According to Corollary \ref{CWGI-C03} and Remark \ref{CWGI-R08}, 
each GI of $\mH$ has the form 
\begin{equation}\label{ISD1-Eq05}
\mX=\mQ\begin{pmatrix} \mX_1 \\ \mX_2 \end{pmatrix}\mP,
\end{equation}
where $\mX_1$ and $\mX_2$ are arbitrary matrices that verify 
$\mV\mX_1+\mX_2=\mId_r$ (please see the notation from (\ref{ISD1-Eq01})).
If we take $\mX_1=0$ and $\mX_2=\mId_r$, 
$x=\mX\vs=\mQ\begin{pmatrix} \mO \\ \mP \vs \end{pmatrix}$ 
is a solution to $\mH\vx=\vs$. This is in fact the Prange $(\mP,\mQ)$-solution.

We can thus say that Prange's algorithm iterates through a particular 
subset of GIs $\mX$ (the matrix $\mQ$ in the 
transformation $(\mP,\mQ)$ is a permutation, $\mX_1 = \mO$, and 
$\mX_2 = \mId_r$) until it finds one with $\Hw{\mX\vs}\leq t$. 
The fact that $\mQ$ is a permutation ensures that the Hamming weight 
of the solution is the Hamming weight of $\mP\vs$. Therefore,
Prange's algorithm is a particular case of our GI-based approach, fact that is formally stated in the following.
\begin{corollary}
Let $\mH\in\cM_{n-k,n}(\bF)$ be a full rank matrix and $\vs\in\bF^{n-k}.$ Prange's algorithm generates solutions to the equation $\mH\vx=\vs$ of the form $\mX\vs$ with
\begin{equation}
    \mX\in\left\{\mQ\begin{pmatrix} \mO \\ \mId_r
                  \end{pmatrix}\mP \mid
              (\mP,\mQ)\in \cGL_r(\bF)\times \cS_n(\bF),\ 
              (\exists\mV:\,\mP\mH\mQ=\begin{pmatrix} 
                                          \mV&  \mId_r
                                       \end{pmatrix})\right\}.
\end{equation}
\end{corollary}

The question now is: does the space of the GIs computed 
by Prange's algorithm cover the entire space $\GI(\mH)$ or not?
This question is justified because Corollary \ref{CWGI-C03} shows us 
how to compute $\GI(\mH)$ using only a transformation $(\mP,\mQ)$
arbitrarily chosen from $\cGL_r(\bF)\times\cGL_n(\bF)$. 
On the other side, Prange's algorithm iterates on all transformations
in $\cGL_r(\bF)\times\cS_n(\bF)$ and, for each transformation
$(\mP,\mQ)$, it computes the generalized inverse 
$\mQ\begin{pmatrix} \mO \\ \mId_r \end {pmatrix}\mP$.

To answer this question we need an intermediate result regarding the 
GIs of a matrix.

\begin{theorem}\label{CWGI-T05}
Let $\mH\in\cM_{n-k,n}(\bF)$ be a full rank matrix. Then, the set 
$\GI(\mH)$ consists of all matrices
$\mX=\mQ\begin{pmatrix}\mO \\ \mId_r\end{pmatrix}\mP$, where 
$(\mP,\mQ)\in \cGL_{n-k}(\bF)\times \cGL_n(\bF)$ is a transformation
of $\mH$ such that 
$\mP\mH\mQ=\begin{pmatrix} \mV & \mId_r \end{pmatrix}$, for some 
$\mV\in\cM_{n-k,k}(\bF)$. 
\end{theorem}
\begin{proof}
Let $\mH\in\cM_{n-k,n}(\bF)$ be a full rank matrix and let $r$ 
denote $n-k$. 

It is straightforward to show that all matrices $\mX$ as in the 
theorem are GIs of $\mH$. 

Let $\mX$ now be a GI of $\mH$.
There exists a transformation $(\mP,\mQ)\in\cGL_r(\bF)\times\cS_n(\bF)$ 
such that $\mP\mH\mQ = \begin{pmatrix} \mV & \mId_r \end{pmatrix}$,
for some matrix $\mV$. According to Remark \ref{CWGI-R08}, there 
are two matrices $\mX_1$ and $\mX_2$ such that   
$\mX = \mQ\begin{pmatrix} \mX_1\\ \mX_2\end{pmatrix} \mP$
and $\mV\mX_1 + \mX_2 = \mId_r$.
We have to show that $\mX$ can be written as 
$\mX = \bar{\mQ}\begin{pmatrix} \mO\\ \mId_r\end{pmatrix} \bar{\mP}$,
for some transformation 
$\bar{\mP}\mH\bar{\mQ}=\begin{pmatrix} \bar{\mV} & \mId_r \end{pmatrix}$.

We prove first that the rank of the matrix 
$\begin{pmatrix} \mX_1 \\ \mX_2 \end{pmatrix}$ is $r$. 
This simply follows from the relation 
  $$\begin{pmatrix} \mV & \mId_r\end{pmatrix}
    \begin{pmatrix} \mX_1 \\ \mX_2 \end{pmatrix} = \mId_r$$
and the well-known rank inequality \cite{Gent2017} 
\begin{equation*}
    \rank(\mA\mB)\leq \min\{\rank(\mA),\rank(\mB)\}.
\end{equation*} 

Since, $\begin{pmatrix} \mX_1 \\ \mX_2 \end{pmatrix}$ is a full column 
rank matrix, we can apply a Gaussian column elimination, which is 
equivalent to computing two operator matrices $\mZ$ and $\mY$ such that 
  $$\begin{pmatrix} \mX_1 \\ \mX_2 \end{pmatrix}=
    \mZ\begin{pmatrix} \mO \\ \mId_r \end{pmatrix} \mY.$$
Now it follows that 
  $$\mX=\mQ\begin{pmatrix} \mX_1 \\ \mX_2 \end{pmatrix} \mP
      =\mQ\mZ\begin{pmatrix} \mO \\ \mId_r \end{pmatrix} \mY\mP
      =\bar{\mQ}\begin{pmatrix} \mO \\ \mId_r \end{pmatrix} \bar{\mP},$$
where $\bar{\mQ}=\mQ\mZ$ and $\bar{\mP}=\mY\mP$.  

It remains to show that there exists a matrix $\bar{\mV}$ such that 
$\bar{\mP}\mH\bar{\mQ}=\begin{pmatrix}\bar{\mV} & \mId_r\end{pmatrix}$.
Decompose $\mZ$ into blocks
$\mZ=\begin{pmatrix} \mZ_1 & \mZ_2 \\ \mZ_3 & \mZ_4 \end{pmatrix}$ 
of appropriate sizes and write:
\begin{align*}
 \bar{\mP}\mH\bar{\mQ} & = \mY(\mP\mH\mQ)\mZ \\
     & = \mY\begin{pmatrix} \mV & \mId_r \end{pmatrix} 
            \begin{pmatrix} \mZ_1 & \mZ_2 \\ \mZ_3 & \mZ_4 \end{pmatrix} \\
     & = \begin{pmatrix} \mY(\mV\mZ_1+\mZ_3) & \mY(\mV\mZ_2+\mZ_4) \end{pmatrix}.
\end{align*} 

But,
  $$\begin{pmatrix} \mX_1 \\ \mX_2 \end{pmatrix}=
    \mZ\begin{pmatrix} \mO \\ \mId_r \end{pmatrix} \mY=
    \begin{pmatrix} \mZ_2\mY \\ \mZ_4\mY \end{pmatrix}$$
and the matrix equation $\mV\mX_1+\mX_2=\mId_r$ leads then to
$\mV\mZ_2+\mZ_4=\mY^{-1}$. This proves that
  $$\bar{\mP}\mH\bar{\mQ} = 
    \begin{pmatrix} \bar{\mV} & \mId_r \end{pmatrix},$$
where $\bar{\mV}=\mY(\mV\mZ_1+\mZ_3)$. 
\end{proof}

In what follows we will discuss several consequences of 
Theorem \ref{CWGI-T05}.
First, it is good to face $\GI(\mH)$ both according to 
Remark \ref{CWGI-R08} and to Theorem \ref{CWGI-T05}:
\begin{itemize}
\item {\em Computing $\GI(\mH)$ by fixing a transformation}: 
  According to Remark \ref{CWGI-R08}, 
      for a given transformation $(\mP,\mQ)\in \cGL_r(\bF)\times \cS_n(\bF)$
      with $\mP\mH\mQ=\begin{pmatrix} \mV & \mId_r\end{pmatrix}$ for some
      $\mV$, we have 
    \[\GI(\mH) = \left\{\mQ\begin{pmatrix} \mX_1 \\ \mX_2\end{pmatrix}\mP 
           \mid \mV\mX_1+\mX_2=\mId_r\right\};\]
\item {\em Computing $\GI(\mH)$ by fixing a GI}: 
  According to Theorem \ref{CWGI-T05}, we have 
  \begin{equation*}
  \GI(\mH) = \left\{\mQ\begin{pmatrix} \mO \\ \mId_r
                  \end{pmatrix}\mP \mid
              (\mP,\mQ)\in \cGL_r(\bF)\times \cGL_n(\bF),\ 
             (\exists\mV:\,\mP\mH\mQ=\begin{pmatrix} 
                                          \mV & \mId_r
                                       \end{pmatrix})\right\}.
\end{equation*}
\end{itemize} 

\begin{remark}\label{CWGI-R10}
In the proof of Theorem \ref{CWGI-T05}, it is not guaranteed that the 
matrix $\mZ$ is a permutation. As a result, $\bar{\mQ}$ may not be a 
permutation. This shows us that Prange's algorithm does not iterate 
through all the GIs of the matrix $\mH$, but only 
on those GIs 
$\mQ\begin{pmatrix} \mO \\ \mId_r
    \end{pmatrix}\mP$
for which $\mQ$ is a permutation. 
\end{remark}

\begin{remark}
Any Prange-like approach that generates solutions to the equation $\mH\vx=\vs$ 
of the form $\mX\vs$ with
\begin{equation}
    \mX\in\left\{\mQ\begin{pmatrix} \mX_1 \\ \mId_r-\mV\mX_1
                  \end{pmatrix}\mP \mid               (\mP,\mQ)\in \cGL_r(\bF)\times \cS_n(\bF),\ 
              (\exists\mV:\,\mP\mH\mQ=\begin{pmatrix} 
                                          \mV&  \mId_r
                                       \end{pmatrix})\right\},
\end{equation}
where $\mX_1$ is fixed, does not cover the set of all solutions to $\mH\vx=\vs$. 
This can be obtained by a similar proof line to Theorem \ref{CWGI-T01}. 
As we shall see in the next subsections, this is the case of many ISD variants, 
such as, Lee-Brickell, Leon, Stern, etc.
\end{remark}

\begin{remark}\label{CWGI-R11}
According to the GI-based framework, Prange's algorithm chooses 
one GI for each transformation it calculates. 
This can be improved in the spirit of Algorithm 2, as follows:
\begin{enumerate}
\item Choose a Prange transformation 
      $\mP\mH\mQ=\begin{pmatrix} \mV & \mId_r\end{pmatrix}$
      to take advantage of the fact that $\mQ$ is a permutation;
\item Iterate through some of the GIs 
        $$\mX=\mQ\begin{pmatrix} \mX_1 \\ \mX_2\end{pmatrix}\mP$$
      with $\mV\mX_1+\mX_2=\mId_r$, until a solution $\mX\vs$ 
      with $\Hw{\mX\vs}\leq t$ is found, or go to step 1 and choose another
      Prange transformation. 
\end{enumerate}
The second step above offers the advantage of computing several 
inverses using the same transformation. The idea is first to choose 
$\mX_1$ and compute $\mX_2$ by $\mX_2=\mId_r-\mV\mX_1$. 
We note that $\mX_1\vs$ must be different from $\vO$ 
because, otherwise, $\mX_2\vs = \mP\vs$, which leads to a Prange 
solution.
\end{remark}

\subsection{Lee-Brickell's approach}\label{ISD2}

As mentioned in Section \ref{ISD0}, a possible attack on McEliece's 
cryptosystem can be mounted guessing the error-free positions. 
An improvement of this attack has been proposed in \cite{LeBr1988}, 
where non-error-free positions are also allowed among the selected 
positions. 
This can equivalently be translated into Prange's formalism, and ours 
too. Details are provided below. In addition, we will develop a 
Lee-Brickell-like method for LWP.

\subsubsection{Lee-Brickell's algorithm}\label{ISD2.1}

In Prange's ISD algorithm, error positions are only allowed outside the 
positions in the information set. With the notation of (\ref{ISD1-Eq04}),
this corresponds to the case $\vz_1=\vO$. 
Lee-Brickell's approach \cite{LeBr1988} allows $\vz_1$ to contain errors,
say at most $p$, where $p$ is a parameter to be chosen depending on 
$n$, $r$, and $t$. 
Once $\vz_1$ is chosen in this way, the equation (\ref{ISD1-Eq04}) 
leads to $\vz_2$. 

For $\mP$ and $\mQ$ given as in (\ref{ISD1-Eq01}), the procedure above is 
iterated for each choice of $p$ positions from the first $k$ positions 
until the solution's Hamming weight is less than or equal to $t$ 
(remark that the matrix $\mQ$ does not change the Hamming weight by 
pre-multiplication).

\subsubsection{Lee-Brickell's approach as a GI-based approach}\label{ISD2.2}

Lee-Brickell's approach is a particular case of the GI-based  
approach. We will discuss two cases here.

\paragraph{Case 1: Lee-Brickell's approach for SDP}

It is trivial to describe Lee-Brickell's approach in terms of 
GIs for SDP. 
In fact, we only have to start from the generalized 
inverse calculated as in Prange's approach and impose the new requirement.
More precisely, given the GI \eqref{ISD1-Eq05} we do as 
follows:
\begin{itemize}
\item Given $p$ positions among the first $k$ positions, we choose the 
   components of $\mX_1$ so that on those $p$ positions, $\mX_1\bar{\vs}$ 
   has only $1$ and $0$ in rest. 
   This always can be done, and even efficiently, when $\bar{\vs}$ has 
   at least one non-zero
   component (please see the proof of Theorem \ref{CWGI-T01});
\item The matrix $\mX_2$ can then be obtained from the equation 
   $\mV\mX_1+\mX_2=\mId_r$.
\end{itemize}

So, Lee-Brickell's approach is just a particular case of our
generic GID solver for CWP (Algorithm \ref{CWGI1-Alg01}).

\paragraph{Case 2: Lee-Brickell's approach for LWP}

As far as we know, Lee-Brickell's approach has not been used 
to find solutions to LWP. However, one can quickly implement 
a Lee-Brickell-like strategy in Algorithm \ref{CWGI2-Alg02} to get 
a solution to this problem. 

Algorithm \ref{CWGI2-Alg02} with transformations as in  
Corollary \ref{CWGI2-C01} outputs solutions of the form 
\begin{equation}
\nonumber
  \vv=\mQ\begin{pmatrix} \mZ \\ -\mV\mZ \end{pmatrix} \mP\vb
\end{equation}
(please see Algorithm \ref{CWGI2-Alg02} for notation). 

As $\vb\not=\v0$, we can easily choose $\mZ$ to get any weight
and distribution of the non-zero components of $\mZ\mP\vb$. 
So, what we have to do is to choose a distribution that 
minimizes the Hamming weight of $-\mV(\mZ\mP\vb)$. We may try
any value of $p$, starting with one up to some reasonable upper
bound.

\subsection{Leon's approach}\label{ISD3} 

In 1988, Leon proposed a probabilistic algorithm for computing 
codewords of weight as small as possible from large error correction 
codes \cite{Leon1988}. 
Without explicitly mentioning or referring to \cite{Pran1962}, 
Leon's algorithm uses the concept of the information set. We will 
describe it in the following.

\subsubsection{Leon's algorithm}\label{ISD3.1}

Let $\mG$ be a generator matrix of an $[n,k]$ linear code for 
which we want to compute a low weight codeword. 
The basic idea of Leon's algorithm is to repeat the following two
steps until a solution is obtained or some halt criterion is fulfilled:
\begin{enumerate}
\item Apply a random permutation $\mQ$ to its columns, followed by row 
  transformations $\mP$, so that the matrix $\mG$ is brought in the form 
  \begin{equation}
  \nonumber 
    \mG'=\mP\mG\mQ=\begin{pmatrix}
           \mId_e & \mZ & \mC \\
            \m0   & \m0 & \mD 
         \end{pmatrix}
  \end{equation}
  where the total length of the first two blocks is $e+\ell$;
\item The following procedure is repeated on 
  $\vv\in\langle\begin{pmatrix} \m0 & \m0 & \mD \end{pmatrix}\rangle$ 
  until the obtained vector has the desired weight, or all $\vv$ 
  were processed:
  \begin{itemize}
  \item Add to $\vv$ zero, one, or two rows from the first $e$ rows 
    but with the requirement that their sum on the positions 
    $\{e+1,\ldots,e+\ell\}$ has weight 0.
  \end{itemize}
\end{enumerate}

Step 2 can be implemented by a systematic enumeration of the space 
$\langle\begin{pmatrix} \m0 & \m0 & \mD \end{pmatrix}\rangle$ 
or by random selection from it. In both cases, the procedure 
starts with a target weight $t$, updated each time a lower weight 
vector is obtained (also saving this vector until the next weight 
update, if necessary).
If the output is the vector $\vv$, the codeword induced by it 
is $\vc=\mQ^{-1}\vv$. 
We notice that the vector $\vv$'s weight on the first $e$ positions 
will always be less than or equal to 2.

Ignoring $\mQ$, which is a permutation and does not change the weight, 
we can say that Leon's algorithm determines codewords (solutions of 
LWP) that have low weight, weight at most two on the 
first $e$, and zero on the following $\ell$ positions. 
Since this vector must be a solution of the equation $\mH\vx = \v0$, 
where $\mH$ is the parity-check matrix of the code generated by $\mG$, 
we deduce that Leon's algorithm 
is looking for low weight solutions of LWP.
These solutions 
have weight at most two on the first $e$, and zero on the 
following $\ell$ positions. Since the transformation from step 1 
is a Prange transformation, we can view this algorithm as a 
variant of Lee-Brickell's algorithm in which the solution is as 
above. 

We can make a step further and generalize Leon's algorithm to the 
SDP. 
To understand how this is done, let us consider $(\mH,\vs,t)$ an
instance of this problem, 
$\mP\mH\mQ=\begin{pmatrix} \mV & \mId_r \end{pmatrix}$ a Prange 
transformation of $\mH$, and $p$ and $\ell$ two positive integers. 

Using the notation in Section \ref{ISD1.1}, we divide the matrices 
$\mV$, $\vz_2$, and $\bar{\vs}$ into two blocks each, 
$\mV=\begin{pmatrix} \mV_1 \\ \mV_2 \end{pmatrix}$, 
$\vz_2=\begin{pmatrix} \vz_2' \\ \vz_2'' \end{pmatrix}$, 
and 
$\bar{\vs}=\begin{pmatrix} \bar{\vs}' \\ \bar{\vs}'' \end{pmatrix}$,
where the first block contains the first $\ell$ rows in each 
and the second block the rest of the rows. 
Then, (\ref{ISD1-Eq04}) can be written as 
\begin{equation}\label{GI-Eq10} 
\begin{pmatrix} \mV_1\vz_1 \\ \mV_2\vz_1 \end{pmatrix} + 
\begin{pmatrix} \vz_2' \\ \vz_2'' \end{pmatrix} = 
\begin{pmatrix} \bar{\vs}' \\ \bar{\vs}'' \end{pmatrix} 
\end{equation}
The requirements of Leon's algorithm are now $\Hw{\vz_1}\leq p$ and  
$\vz_2'=0$. Remark also that $\vz_1$ must additionally satisfy  
$\mV_1\vz_1=\bar{\vs}'$. Assuming that all these are fulfilled,  
$\vz_2''$ is determined from $\mV_2\vz_1+\vz_2''=\bar{\vs}''$.

\subsubsection{Leon's approach as a GI-based approach}\label{ISD3.2}

Leon's approach is a particular case of the GI-based 
approach. We will discuss two cases here.

\paragraph{Case 1: Leon's approach for LWP}

Let $(\mH,t)$ be an instance of LWP. 
Algorithm \ref{CWGI2-Alg02} with transformations as in  
Corollary \ref{CWGI2-C01} outputs solutions of the form 
\begin{equation}
\nonumber
  \vv=\mQ\begin{pmatrix} \mZ \\ -\mV\mZ \end{pmatrix} \mP\vb
\end{equation}
(please see Algorithm \ref{CWGI2-Alg02} for notation). 
We further decompose $\mV\mZ$ into two parts 
\begin{equation}
\nonumber
  \vv=\mQ\begin{pmatrix} \mZ \\ -\mV_1\mZ \\ -\mV_2\mZ \end{pmatrix} \mP\vb
\end{equation}
where $\mV_1$ contains the first $\ell$ rows and $\mV_2$
the other rows of $\mV$. 

As $\vb\not=\v0$, we can easily choose $\mZ$ to get any distribution
of the non-zero components of $\mZ\mP\vb$. 
As a result, Algorithm \ref{CWGI2-Alg02} iterates through Prange
transformations of $\mH$, and for each transformation, iterates through 
some values of $\mZ$ until it finds one with the properties 
$\Hw{\mZ\mP\vb}\leq p$, $\Hw{-\mV_1\mZ\mP\vb}=\v0$, and $\Hw{\vv}\leq t$
(please also see the comments at the end of Section \ref{CWGI1.4}).

\paragraph{Case 2: Leon's approach for SDP}

Let $(\mH,\vs,t)$ be an instance of SDP. 
Algorithm \ref{CWGI1-Alg01} will iterate through Prange transformations
$\mP\mH\mQ=\begin{pmatrix} \mV & \mId_r \end{pmatrix}$, and for each
transformation it will iterate through GIs of the form 
\begin{equation}\label{ISD3.2-Eq05}
\mX=\mQ\begin{pmatrix} \mX_1 \\ \mX_2 \end{pmatrix}\mP,
\end{equation}
where $\mX_1$ and $\mX_2$ are arbitrary matrices that verify 
$\mV\mX_1+\mX_2=\mId_r$.

We further decompose $\mV\mX_1+\mX_2=\mId_r$ as 
\begin{equation}\label{GI-Eq09}
  \nonumber 
  \begin{pmatrix} \mV_1\mX_1 \\ \mV_2\mX_1 \end{pmatrix} + 
  \begin{pmatrix} \mX_2' \\ \mX_2'' \end{pmatrix} = 
           \begin{pmatrix} \mId_r' \\ \mId_r'' \end{pmatrix} 
\end{equation}
  where $\mV_1$, $\mX_2'$, and $\mId_r'$ contain the first $\ell$ rows
  of $\mV$, $\mX_2$, and $\mId_r$, respectively, while  
  $\mV_2$, $\mX_2''$, and $\mId_r''$ contain the other rows. 

Now, what we have to do is to compute $\mX_1$, $\mX_2'$, and 
$\mX_2''$ under the supplementary requirements $\mX_2'\bar{\vs}=\m0$ 
and $\Hw{\mX_1\bar{\vs}}\leq p$.
This can be implemented in various ways. 
For example, assuming $\bar{\vs}\not = \v0$, we can choose $\mX_1$ 
so that $\mX_1\bar{\vs}$ is any $(n-r)$-dimensional vector
(see, for instance, the proof of Theorem \ref{CWGI-T01}). 
So suffice it to say that this is a vector with at most $p$ positions 1. 
Then, for each such possibility, we check whether the sum of $\mV_1$'s 
columns corresponding to positions 1 in 
$\mX_1\bar{\vs}$ gives us $\bar{\vs}'$ (see the notation in 
Section \ref{ISD3.1}). If so, we keep $\mX_1\bar{\vs}$ and 
compute $\mX_2''$. This is, for instance, the method in \cite{BBCPS2019}. Tacking in both cases (Leon's approach for SDP/LWP) $p=2$ gives the original version proposed by Leon.

\subsection{Stern's approach}\label{ISD4}

Independent and almost at the same time with Leon, Stern proposed an 
algorithm to determine low-weight codewords \cite{Ster1988} that use 
the concept of the information set. 
We will describe it below and show that it fits very well in the 
GI-based approach.

\subsubsection{Stern's algorithm}\label{ISD4.1}

Unlike Leon, the codewords in Stern's algorithm are determined 
by the parity-check matrix $\mH\in\cM_{n-k,n}(\bF_2)$, so as solutions 
of the equation $\mH\vx=\m0$.
To determine a zero linear combination of $\mH$'s columns, 
Stern proposed the following algorithm that uses two parameters $p$ and
$\ell$ that are optimized later (please see Figure \ref{Fig-Stern}
for an illustration of the algorithm, and consider $\bar{\vs}=\v0$):
\begin{enumerate}
\item Select $r=n-k$ columns from $\mH$ and, by row elementary 
  transformations, form the columns of the unit matrix.
  For readability, we assume that these columns are packed to the 
  right because the columns permutation does not change the solution's 
  weight (we point out that this requirement is not present in 
  Stern's algorithm). 
  As a result, a Prange transformation is applied to $\mH$,
  $\mP\mH\mQ=\begin{pmatrix} \mV & \mId_r \end{pmatrix}$.
  Then, $\vx$ is a solution to $\mH\vx=\v0$ if and only if 
  $\vz$ is a solution to 
  $\begin{pmatrix} \mV & \mId_r \end{pmatrix}\vz=\v0$, where 
  $\vz=\mQ^{-1}\vx$; 
\item The rest of the columns are randomly distributed 
  in two sets of approximately equal size. Let $\cI_1$ and $\cI_2$ be the 
  index sets of these column sets. 
  With the notation from the previous step, we can say that $\mV$ is 
  divided into two blocks of approximately the same size,
  $\mV=\begin{pmatrix} \mV_1 & \mV_2 \end{pmatrix}$;
\item Randomly choose a set $\cL$ of $\ell$ row position; 
\item Choose $\cJ_1\subseteq \cI_1$ and $\cJ_2\subseteq \cI_2$ of equal size $p$ 
  such that the sum of the columns with index in $\cJ_1$ is equal 
  to the sum of the columns with index in $\cJ_2$ on the positions in $\cL$. 
  This will ensure that the sum of all these columns is 0 on all 
  positions in $\cL$;
\item On certain positions outside the set $\cL$, the sum of the 
  columns with the index in $\cJ_1\cup \cJ_2$ can be 1. For each such 
  position $j$, the column from the unit matrix will be added with 
  1 on position $j$. If we denote by $\cJ_3$ the set of all these 
  positions, the sum of the columns with the index in 
  $\cJ_1\cup \cJ_2\cup \cJ_3$ will be 0;
\item So the vector that has 1 only on the positions in 
  $\cJ_1\cup \cJ_2\cup \cJ_3$, and otherwise only zero will be a
  codeword with weight $2p+|\cJ_3|$.
\end{enumerate}

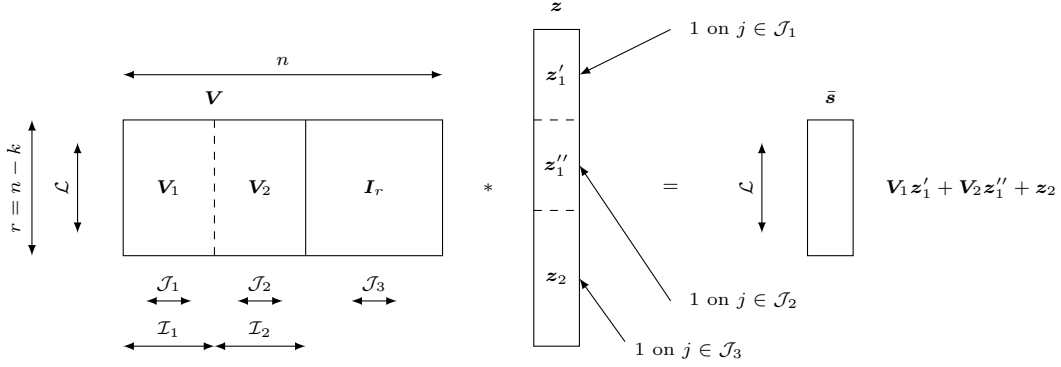
\begin{figure}[!ht]
\centering
\begin{tikzpicture}[scale=1.2,font=\scriptsize]
\draw  (-4,2.5) rectangle (-0.5,1);
\draw[latex-latex] (-5,1) -- node[sloped,above] {$r=n-k$} (-5,2.5);
\draw[latex-latex] (-4,3) -- node[above] {$n$} (-0.5,3);
\draw[dashed] (-3,2.5) -- (-3,1);
\draw (-2,2.5) -- (-2,1);
\node at (-3.5,1.75) {$\mV_1$};
\node at (-2.5,1.75) {$\mV_2$};
\node at (-1.25,1.75) {$\mId_r$};
\draw[latex-latex] (-3.75,0.5) -- node[above] {${\mathcal J}_1$} (-3.25,0.5);
\draw[latex-latex] (-2.75,0.5) -- node[above] {${\mathcal J}_2$} (-2.25,0.5);
\draw[latex-latex] (-1.5,0.5) -- node[above] {${\mathcal J}_3$} (-1,0.5);
\draw[latex-latex] (-4,0) -- node[above] {${\mathcal I}_1$} (-3,0);
\draw[latex-latex] (-3,0) -- node[above] {${\mathcal I}_2$} (-2,0);
\draw  (0.5,3.5) rectangle (1,0);
\draw[dashed] (0.5,2.5) -- (1,2.5);
\draw[dashed] (0.5,1.5) -- (1,1.5);
\draw[latex-latex] (-4.5,1.25) -- node[sloped,above] {${\mathcal L}$} (-4.5,2.25);
\draw[latex-latex] (3,1.25)  -- node[sloped,above] {${\mathcal L}$} (3,2.25);
\draw  (3.5,2.5) rectangle (4,1);
\node[right] at (4.25,1.75) {$\mV_1\vz_1'+\mV_2\vz_1''+\vz_2$};
\node at (0.75,3) {$\vz_1'$};
\node at (0.75,2) {$\vz_1''$};
\node at (0.75,0.75) {$\vz_2$};
\node at (0,1.75) {$*$};
\node at (2,1.75) {$=$};
\draw[latex-] (1,3) -- (2,3.5);
\draw[latex-] (1,2) -- (2,0.5);
\node[right] at (2.1,0.5) {1 on $j\in {\mathcal J}_2$};
\node[right] at (2.1,3.5) {1 on $j\in {\mathcal J}_1$};
\draw[latex-] (1,0.75) -- (1.5,0.05);
\node[right] at (1.5,-0.05) {1 on $j\in {\mathcal J}_3$};
\node at (3.75,2.75) {$\bar{\vs}$};
\node at (0.75,3.75) {$\vz$};
\node at (-3,2.75) {$\mV$};
\end{tikzpicture}
\caption{Stern's algorithm illustrated}\label{Fig-Stern}
\end{figure}

Stern's algorithm is a solution to LWP. We can trivially 
extend it to be a solution to SD instance $\mH\vx=\vs$. 
The equation we have to solve now is 
$\begin{pmatrix} \mV & \mId_r \end{pmatrix}\vz=\bar{\vs}$, where
$\bar{\vs}=\mP\vs$
(please see the diagram in Figure \ref{Fig-Stern}). 
This implies that on the positions in $\cL$ we must have satisfied 
the relation
\begin{equation}\label{ISD4.1-Eq01}
(\mV_1\vz_1')_{\cL} + (\mV_2\vz_1'')_{\cL}=\bar{s}_{\cL}
\end{equation}

Choosing $\vz_1'$ and $z_1''$ in this way, $\vz_2$ is obtained from 
\begin{equation}\label{ISD4.1-Eq02}
\mV_1\vz_1'+\mV_2\vz_1''+\vz_2=\bar{s}
\end{equation}
It should be noted that $\vz_2$ will have one on the positions in $\cJ_3$.
The solution $\vx$ is obtained by $\vx=\mQ\vz$.

The computation of $\vz_1'$ and $\vz_1''$, which satisfies equation 
\eqref{ISD4.1-Eq01}, can be done this way.
For each subset $\cJ_1\subseteq\cI_1$ of size $p$, the value 
$\mV_1\vz_1'+\bar{s}$ is stored in a list $L$, where $\vz_1'$ 
has $1$ on all positions in $\cJ_1$, and 0 in rest. 
Then, whenever it is necessary to find $\vz_1''$, a subset 
$\cJ_2\subseteq\cI_2$ of size $p$ is chosen.
If the equation \eqref{ISD4.1-Eq01} is fulfilled, then $\vz_2$
is obtained from \eqref{ISD4.1-Eq02}.

\subsubsection{Stern's approach as a GI-based approach}\label{ISD4.2}

Stern's approach is a particular case of the GI-based 
approach. We will discuss two cases here.

\paragraph{Case 1: Stern's approach for LWP}

Algorithm \ref{CWGI2-Alg02} with the transformations from 
Corollary \ref{CWGI2-C01} will output solutions for LWP
of the form 
\begin{equation}
\nonumber
  \vv=\mQ\begin{pmatrix} \mZ \\ -\mV\mZ \end{pmatrix} \mP\vb
\end{equation}
(please see Algorithm \ref{CWGI2-Alg02} for notation). 
As $\vb\not=\v0$, we can easily choose $\mZ$ to get any distribution
on the non-zero components of $\mZ\mP\vb$. This corresponds to $\vz_1$
in Stern's original approach. The second component, $-\mV\mZ$, 
corresponds to $\vz_2$ in Stern's approach
(please also see the comments at the end of Section \ref{CWGI1.4}). 

So, Algorithm \ref{CWGI2-Alg02} in this case will mainly iterate 
through $\mZ$'s values and transformations $(\mP,\mQ)$. 

\paragraph{Case 2: Stern's approach for SDP}

Stern's approach for SDP is a 
particular case of GI-based approach (Algorithm \ref{CWGI1-Alg01}).  
To see that, consider the transformation \eqref{ISD1-Eq01}.
Each GI of $\mH$ has the form 
\begin{equation}\label{ISD4.2-Eq05}
\mX=\mQ\begin{pmatrix} \mX_1 \\ \mX_2 \end{pmatrix}\mP,
\end{equation}
where $\mX_1$ and $\mX_2$ are arbitrary matrices that verify 
$\mV\mX_1+\mX_2=\mId_r$ (please see the notation from (\ref{ISD1-Eq01})).
We further decompose the matrices $\mV$ and $\mX_1$ into two parts
each, $\mV=\begin{pmatrix} \mV_1 & \mV_2 \end{pmatrix}$ and
$\mX_1=\begin{pmatrix} \mX_1' \\ \mX_1'' \end{pmatrix}$. Then, we
arrive at the equation
\begin{equation}
 \mV_1\mX_1' + \mV_2\mX_1'' + \mX_2=\mId_r
\end{equation}

The choice of matrices $\mX_1'$ and $\mX_1''$ can be made on 
principles similar to those in the previous section.

\subsection{Finiasz-Sendrier's approach} 

In \cite{FiSe2009}, Finiasz and Sendrier propose a ``generalization'' of 
Stern's approach, making a partial Gaussian elimination 
on the matrix. We will discuss this approach below.

\subsubsection{Finiasz-Sendrier's algorithm}

The transformation used on $\mH$ in Finiasz-Sendrier's approach is
\begin{equation}\label{GI-Eq11}
  \mP\mH\mQ=\begin{pmatrix} \mV_1 & \m0 \\ \mV_3 & \mId_{r-\ell} \end{pmatrix}, 
\end{equation}
where $\mV_1\in\cM_{\ell,k+\ell}(\bF)$, $\mV_2\in\cM_{r-\ell,k+\ell}(\bF)$,
and $\ell\geq 0$.  

The equation $\mP\mH\mQ\vz=\mP\vs$ leads then to
\begin{equation}\label{GI-Eq12} 
\left\{
\begin{array}{lcl}
  \mV_1\vz_1 &=& \bar{\vs}_1 \\
  \mV_3\vz_1 + \vz_2 &=& \bar{\vs}_2
\end{array}
\right.
\end{equation}
where $\vz=\begin{pmatrix} \vz_1 \\ \vz_2 \end{pmatrix}$, 
$\vz_1\in\bF^{k+\ell}$, $\vz_2\in\bF^{r-\ell}$, 
$\bar{\vs}=\mP\vs=\begin{pmatrix} \bar{\vs}_1 \\ \bar{\vs}_2 \end{pmatrix}$, 
$\bar{\vs}_1\in\bF^{\ell}$, and $\bar{\vs}_2\in\bF^{r-\ell}$. 

The vector $\vz_1$ is divided into two parts, $\vz_1'$ and $\vz_1''$, that
are computed as in Stern's approach, while $\vz_2$ is obtained from the 
second equation in (\ref{GI-Eq12}). At the end, $\vx=\mQ\vz$.

\subsubsection{Finiasz-Sendrier's approach and a GI-based approach.}

As with the other approaches, Finiasz-Sendrier's approach reduces 
to a GI calculation. Any GI of 
the matrix in (\ref{GI-Eq11}) has the form
$\begin{pmatrix} \mX_1 & \mX_2 \\ \mX_3 & \mX_4 \end{pmatrix}$, where
$\mX_1\in\cM_{k+\ell,\ell}(\bF)$, $\mX_2\in\cM_{k+\ell,r-\ell}(\bF)$,
$\mX_3\in\cM_{r-\ell,\ell}(\bF)$, $\mX_4\in\cM_{r-\ell,r-\ell}(\bF)$,
and the following properties hold:
\begin{equation}\label{GI-Eq13} 
\begin{cases}
  \mV_1\mX_1\mV_1         &= \mV_1 \\
  \mV_1\mX_2              &= \m0 \\
  (\mV_3\mX_1+\mX_3)\mV_1 &= \m0 \\
  \mV_3\mX_2+\mX_4        &= \mId_{r-\ell}
\end{cases}
\end{equation}

The general solution to the system $\mP\mH\mQ\vz=\mP\vs$ will then be 
\begin{equation}\label{GI-Eq14}
  \vz=\begin{pmatrix} \mX_1 & \mX_2 \\ \mX_3 & \mX_4 \end{pmatrix}
     =\begin{pmatrix} \mX_1\bar{\vs}_1+\mX_2\bar{\vs}_2 \\
                    \mX_3\bar{\vs}_1+\mX_4\bar{\vs}_2 \end{pmatrix}
\end{equation}
One can easily that $\vz_1$ and $\vz_2$ derived as in (\ref{GI-Eq14}) also
satisfy (\ref{GI-Eq12}). Indeed, first remark that, according to
(\ref{GI-Eq13}), $\mX_1$ is a GI of $\mV_1$ and so
$\mX_1\bar{\vs}_1$ is a solution to $\mV_1\vz_1=\bar{\vs}_1$. Then, 
$$
\begin{array}{lcl}
\mV_1\vz_1 &=& \mV_1(\mX_1\bar{\vs}_1+\mX_2\bar{\vs}_2) \\ [.5ex]
           &=& \mV_1\mX_1\bar{\vs}_1 + \mV_1\mX_2\bar{\vs}_2 \\ [.5ex]
           &=& \bar{\vs}_1
\end{array}
$$
and 
$$
\begin{array}{lcl}
\mV_3\vz_1+\vz_2 &=& \mV_2(\mX_1\bar{\vs}_1+\mX_2\bar{\vs}_2)+
                     \mX_3\bar{\vs}_1+\mX_4\bar{\vs}_2 \\ [.5ex]
   &=& (\mV_3\mX_1+\mX_3)\bar{\vs}_1+(\mV_3\mX_2+\mX_4)\bar{\vs}_2 \\ [.5ex]
   &=& (\mV_3\mX_1+\mX_3)\mV_1\mX_1\bar{\vs}_1 +\mId_{r-\ell}\bar{\vs}_2 \\ [.5ex]
   &=& \m0\mX_1\bar{\vs}_1+\bar{\vs}_2 \\ [.5ex]
   &=& \bar{\vs}_2.
\end{array}
$$

Several strategies of selecting matrices in \eqref{CWGI-Eq13} could be imagined. 
For example, setting $\mX_2=\mO$ and $\mX_3=-\mV_3\mX_1$ implies $\mX_4=\mId_{r-l}$ 
and so, the solution has the form 
\begin{equation}\label{GI-Eq14b}
  \vz=\begin{pmatrix} \mX_1\bar{\vs}_1 \\
                   -\mV_3\mX_1\bar{\vs}_1+\bar{\vs}_2 \end{pmatrix}
\end{equation}
In this case the following procedure could be applied:
\begin{itemize}
    \item Sample $\mX_1\in\GI(\mV_1)$ s.t. $w=\Hw{\mX_1\bar{\vs_1}}\leq t$;
    \item If $\Hw{-\mV_3\mX_1\bar{\vs}_1+\bar{\vs}_2}\leq t-w$ then output $\vz$. 
    Else, go to the previous step.
\end{itemize}

\subsection{Multiple decompositions}\label{ISD6}

Let $(\mH,\vs,t)$ be an instance of SDP and 
$\begin{pmatrix} \mV & \mId_r \end{pmatrix}$ a Prange transformation
of $\mH$. 
If we look at the GI-based approach of the Leon and Finiasz-Sendrier's algorithms, 
we find that $\mV$ is vertically decomposed into two other blocks. The GI-based 
approach of Stern's algorithm requires the horizontal decomposition of $\mV$ to 
do a ``meet-in-the-middle'' (``birthday matching'') on sub-blocks.
We can combine these two decompositions as follows. 

Let 
\begin{equation}\label{ISD6-Eq01}
  \mP\mH\mQ=\begin{pmatrix} \mV_1 & \mId_{\ell} & \m0 \\
                            \mV_2 & \mO      & \mId_{r-\ell}
    \end{pmatrix},
\end{equation}
where $\ell$ is a positive integer. 
It is straightforward to see that any GI of $\mH$ has
the form
\begin{equation}\label{ISD6-Eq02}
\mX=\mQ\begin{pmatrix}
        \mX_1                  & \mX_2\\
        \mId_{\ell}-\mV_1\mX_1 & -\mV_1\mX_2\\
        -\mV_2\mX_1            & \mId_{r-\ell}-\mV_2\mX_2
    \end{pmatrix}\mP, 
\end{equation}
where $\mX_1\in\cM_{k,\ell}(\bF)$ and $\mX_2\in\cM_{k,r-\ell}(\bF)$ 
are matrices of appropriate sizes. 
A solution will be then of the form 
\begin{equation}\label{ISD6-Eq03}
\mX\vs=\mQ\begin{pmatrix}
        \mX_1\bar{\vs}'             + \mX_2\bar{\vs}'' \\
        (\bar{\vs}'-\mV_1\mX_1\bar{\vs}') + (-\mV_1\mX_2\bar{\vs}'') \\
        (-\mV_2\mX_1\bar{\vs}')     + (\bar{\vs}''-\mV_2\mX_2\bar{\vs}'')
    \end{pmatrix}, 
\end{equation}
where $\bar{\vs}=\mP\vs$, $\bar{\vs}'=\bar{\vs}_{[1,\ell]}$, and 
$\bar{\vs}''=\bar{\vs}_{[\ell+1,r-\ell]}$.

Assuming that $\bar{\vs}'\not=\v0$ and 
$\bar{\vs}''\not=\v0$, $\mX_1\bar{\vs}'$ and $\mX_2\bar{\vs}''$ may
independently define any linear combinations $\mV_1$'s columns.
Therefore, the middle horizontal row in $\mX\vs$ is a generalization
of the corresponding part in Leon and Stern's algorithms. 
In addition, no higher computational costs are required than in these 
algorithms and leave a door open to other approaches.

We can generalize even more. For instance, the transformation 
\begin{equation}\label{ISD6-Eq04}
\mP\mH\mQ=\begin{pmatrix}
        \mV_1  & \mId_{\ell_1}  & \m0        & \m0 \\
        \mV_2  & \m0         & \mId_{\ell_2} & \m0 \\
        \mV_3  & \m0         & \m0        & \mId_{r-(\ell_1+\ell_2)} 
          \end{pmatrix}
\end{equation}
leads to GIs of the form 
\begin{equation}\label{ISD6-Eq05}
\mX=\mQ\begin{pmatrix}
        \mX_1&\mX_2&\mX_3\\
        \mId_{\ell_1}-\mV_1\mX_1&-\mV_1\mX_2&-\mV_1\mX_3\\
        -\mV_2\mX_1&\mId_{\ell_2}-\mV_2\mX_2&-\mV_2\mX_3\\
        -\mV_3\mX_1&-\mV_3\mX_2&\mId_{r-(\ell_1+\ell_2)}-\mV_3\mX_3
        \end{pmatrix}\mP
\end{equation}

If we decompose $\bar{\vs}$ into three corresponding parts,
$\bar{\vs}'$, $\bar{\vs}''$, and $\bar{\vs}'''$, and assume that each of
them is different than $\v0$, then 
$\mX_1\bar{\vs}'$, $\mX_2\bar{\vs}''$, and $\mX_3\bar{\vs}'''$
may independently define any linear combination of $\mV_1$, 
$\mV_2$, and $\mV_3$'s columns.

\section{GI-based decoding and Boolean constraint satisfaction}\label{GID-CSP}

This section will look at the CW and SWPs as optimization problems 
and link them to the MIN-SAT problem. To have a clearer picture of 
these problems, we recall first the concept of an optimization problem. 
For details, the reader is referred to \cite{ACGKMP2003, CrKS2001}.

\begin{definition}\label{def:opt-pb}
An {\em optimization problem} is a tuple $A=(\cI,Sol,\mu,Opt)$,
where:
\begin{enumerate}
\item $\cI$ is a set of {\em instances};
\item $Sol$ is a function that maps each instance $x$ to a set
	of feasible solutions of $x$;
\item $\mu$ is a {\em measure function} that associates to each pair
 	$(x,y)\in\cI\times Sol(x)$ a positive integer regarded as a measure 
 	of the instance $x$'s feasible solution $y$;
\item $Opt\in\{min,max\}$ is the {\em optimization criterion} (minimization
  	or maximization). 
\end{enumerate}
\end{definition}

Given an optimization problem as above, we denote by $Sol^*(x)$ the
set of optimal solutions of $x$ (that is, those feasible solutions 
from $Sol(x)$ that fulfill the optimization criterion $Opt$). 
Moreover, $\mu^*(x)$ stands for the value of the optimal solutions
of $x$, when they exist. 

Each optimization problem $A$ defines three related problems:
\begin{enumerate}
\item {\em Constructive problem} $A_C$: Given an instance $x\in\cI$, 
  compute an optimal solution and its measure;
\item {\em Evaluation problem} $A_E$: Given an instance $x\in\cI$,
  compute $\mu^*(x)$;
\item {\em Decision problem} $A_D$: Given an instance $x\in\cI$ and 
  a positive integer $t$, decide the predicate
  \begin{equation}
    A_D(x,t)=\left\{
               \begin{array}{ll}
                 1, & (Opt=min\ \wedge\ \mu^*(x)\leq t)\ \vee\ \\ [.5ex]
                    & (Opt=max\ \wedge\ \mu^*(x)\geq t) \\ [.5ex] 
                 0, & \text{otherwise}
               \end{array}
             \right.
  \end{equation}
\end{enumerate}
The {\em underlying language} of $A$ is the set
$\{(x,t)\in\cI\times\bN\mid A_D(x,t)=1\}$. 

\begin{definition}\label{def:NPO}
An optimization problem $A=(\cI,Sol,\mu,Opt)$ belongs to the class
NPO if:
\begin{enumerate}
\item Its set of instances is recognizable in polynomial time;
\item There exists a polynomial $p$ such that for each instance $x$,
  	its feasible solutions have size at most $p(|x|)$. Moreover,
  	for each $y$ of size at most $p(|x|)$, it is decidable in polynomial
  	time whether $y$ is a feasible solution of $x$;
\item The measure function is computable in polynomial time. 
\end{enumerate}
\end{definition}

Now is the time to introduce the MIN-CWP and MIN-SWP.

\begin{problem}[framed]{MIN Coset Weight Problem  (MIN-CWP)}
Instance: & Matrix $\mA\in\cM_{m,n}(\bF)$ with full row rank
            and vector $\vb\in\bF^m$, 
            where $\bF$ is a finite field; \\
Sol:      & Solution $\vx_0\in\bF^n$ to $\mA\vx=\vb$; \\
Measure:  & Hamming weight of $\vx_0$; \\
Opt:      & $min$.  
\end{problem}

If we choose $\vb=\v0$ in MIN-CWP we get the MIN-SWP. 
We distinguish between these problems because the methods for 
solving them are different (see Section \ref{CWGI1} for the CWP and Section \ref{CWGI2} for SWP).

Remark that the matrix $\mA$ in MIN-CWP and MIN-SWP has full row rank. That 
is particularly important for the rest of this section from several 
points of view: we may use Theorems \ref{CWGI1-T01} and \ref{CWGI2-T01} 
to characterize the set of feasible solutions of an instance, and we 
have a suitable form of the GIs. On the other hand, this 
also corresponds to SDP and LWPs. 
Of course, one may extend the framework if necessary.

\begin{proposition}\label{GID-CSP-P01}
The following properties hold:
\begin{enumerate}
\item MIN-CWP and MIN-SWP are NPO problems.
\item CWP (SWP) is the decision problem of MIN-CWP (MIN-SWP).
\item The constructive, evaluation, and the decision problems of
  	MIN-CWP (MIN-SWP) are all polynomial-time Turing-equivalent. 
\end{enumerate}
\end{proposition}
\begin{proof}
We will prove the theorem in the case of the MIN-CWP.

(1) It is quite clear that each MIN-CWP instance $(\mA,\vb)$ is recognizable
in polynomial time. Given an instance $(\mA,\vb)$, its set of feasible
solutions consists of all solutions to the equation $\mA\vx=\vb$. 
According to Theorem \ref{CWGI1-T01}, this set is
  $$Sol(\mA,\vb)=\{\mX\vb\mid \mX\in\GI(\mA)\}.$$
Clearly, each feasible solution is linear in the size of the instance.
Moreover, given a vector $\vc\in\bF^n$, we can decide in polynomial time
if it is a solution to $\mA\vx=\vb$ (that is, whether it is a feasible
solution or not of the instance $(\mA,\vb)$). Finally, the measure of a
feasible solution can be computed in linear time with respect to the size
of the solution.

(2) follows directly from definitions.

(3) We know that CWP is an NP-complete problem \cite{BeET1978}. According
to Theorem 1.5 in \cite{ACGKMP2003}, all the problems associated to an 
NPO problem are polynomial-time Turing-equivalent if the decision problem 
is NP-complete. So, we get the result. 
\end{proof}

We will establish further a helpful connection between the MIN-CWP 
and MIN-SWPs over $\bF_2$, denoted MIN-CWP($\bF_2$) and 
MIN-SWP($\bF_2$), and the constraints satisfaction framework. 
Let us first recall the classical constraint satisfaction framework 
\cite{CrKS2001}. 

Let $D$ be a finite set of cardinality at least two that we consider 
fixed during this section.
A {\em constraint} over a finite set $\cV$ of variables is a pair 
$\sigma=(f,(x_{1},\ldots,x_{\ell}))$ consisting of a function 
$f:D^\ell\ra\{0,1\}$, for some $\ell\geq 1$, and a list 
$(x_{1},\ldots,x_{\ell})$ of variables from $\cV$ that 
take values in $D$. 
The function $f$ is called the {\em constraint function}, and the
list of variables, the {\em constraint scope} of $\sigma$. 
The variables in the constraint scope may not be pairwise distinct. 

An {\em instance of the constraint satisfaction problem} over a set
$\cV$ of variables is a finite set $\Sigma$ of constraints over $\cV$.
Distinct constraints in $\Sigma$ may have distinct constraint functions 
or scopes.

Let $\gamma:\cV\ra D$ be an assignment of the variables in $\cV$.
A constraint $\sigma=(f,(x_{1},\ldots,x_{\ell}))$ over $\cV$ is
{\em satisfied} by $\gamma$ is $f(\gamma(x_1),\ldots,\gamma(x_\ell))=1$.

Now, we are able to formulate the following problem. 

\begin{problem}[framed]{MIN Constraint Satisfaction Problem (MIN-CSP)}
Instance: & Finite set $\Sigma$ of constraints over a finite set $\cV$
            of variables; \\
Goal:     & Assignment $\gamma:\cV\ra D$; \\
Measure:  & Number of constraints in $\Sigma$ satisfied by  $\gamma$; \\
Opt:      & $min$.  
\end{problem}

When $D=\{0,1\}$ and each constraint function is a Boolean function,
MIN-CSP is usually referred to as the MIN-SAT problem. When the
constraints are affine, that is, expressions of the form
   $$x_1\oplus\cdots\oplus x_n=b,$$
where $n\geq 1$, $x_1,\cdots,x_n\in\cV$, and $b\in\{0,1\}$, the problem 
is usually denoted MIN-SAT(affine). 

We will show that MIN-CWP($\bF_2$) and MIN-SWP($\bF_2$) are reducible 
in a very strict way to the MIN-SAT problem. Recall first the concept
of reducibility between optimization problems
\cite{ACGKMP2003,Cres1997}.  

\begin{definition}\label{def:reduction}
Let $A=(\cI_A,Sol_A,\mu_A,Opt_A)$ and $B=(\cI_B,Sol_B,\mu_B,Opt_B)$ 
be two NPO problems. A {\em reduction} from $A$ to $B$ is a pair 
$(f,g)$ of polynomial-time computable function such that:
\begin{enumerate}
\item $f(x)\in\cI_B$, for any $x\in\cI_A$;
\item $g(x,y)\in Sol_A(x)$, for any $x\in\cI_A$ and $y\in Sol_B(f(x))$. 
\end{enumerate}
\end{definition}
\begin{definition}\label{def:S-reduction}
Let $A,B$ be two NPO problems and $(f,g)$ be a reduction from $A$ to $B$. Then $(f,g)$ is called an {\em $S$-reduction} from $A$ to $B$ if the following two properties hold:
\begin{enumerate}
\item $\mu_B^*(f(x))=\mu_A^*(x)$, for any $x\in\cI_A$;
\item $\mu_A(x,g(x,y))=\mu_B(f(x),y)$, for any $x\in\cI_A$ and 
  	$y\in Sol_B(f(x))$. 
\end{enumerate}
\end{definition}

\begin{theorem}\label{GID-CSP-T01}
MIN-CWP$(\bF_2)$ and MIN-SWP$(\bF_2)$ are $S$-reducible to MIN-SAT(affine).
\end{theorem}
\begin{proof}
We will discuss first the case of the MIN-CWP$(\bF_2)$ problem. 

Let $I=(\mA,\vb)$ be an instance of the MIN-CWP$(\bF_2)$ problem, where 
$\mA\in\cM_{m,n}(\bF_2)$, $\vb\in\bF_2^m$, $\vb\not=\v0$, and 
$rank(\mA)=m<n$. 
We associate to $I$ the following instance of MIN-SAT(affine). 

First, by Gaussian elimination, transform $\mA$ into 
$\mP\mA\mQ=\begin{pmatrix} \mA_1 & \mId_m \end{pmatrix}$, 
for some matrices $\mP\in \cGL_m(\bF)$, $\mQ\in \cS_n(\bF)$, and $\mA_1$. 
Consider then a set $\cV=\{z_1,\ldots,z_{n-m}\}$ of $n-m$ Boolean 
variables and define the set $\Sigma$ of affine constraints
\begin{equation}
\Sigma:\ \left\{
\begin{array}{lcl}
  z_1 &=& 1 \\
      &\cdots& \\
  z_{n-m} &=& 1 \\
  \mP\vb - \mA_1\vz &=& \bm{1}, 
\end{array} 
\right.
\end{equation}
where $\vz=\begin{pmatrix} z_1 \\ \cdots \\ z_{n-m} \end{pmatrix}$
and $\bm{1}$ is a vector of 1's.
Each row in $\mP\vb - \mA_1\vz=\bm{1}$ is a constraint, and so $\Sigma$
contains exactly $n$ constraints. 
Thus, $(\cV,\Sigma)$ is an instance of the MIN-SAT(affine) problem,  
which can be constructed in polynomial-time from $I$. 

So, we have defined a polynomial-time computable function $f$ that,
on an instance $I$ of MIN-CWP$(\bF_2)$, associates an instance
$f(I)=(\cV,\Sigma)$ of MIN-SAT(affine). 

Let $\gamma\in Sol_B(\cV,\Sigma)$ be a feasible solution for $(\cV,\Sigma)$. 
Extend $\gamma$ homomorphically and component-wise to affine expressions 
and systems of affine expressions. Let $\bar{\gamma}$ be the extension. 
Define now $\vx_\gamma\in\bF^n$ as follows:
\begin{itemize}
\item $(\vx_\gamma)_{[1,n-m]}=\bar{\gamma}(\vz)\in\bF_2^{n-m}$;
\item $(\vx_\gamma)_{[n-m+1,n]}=\bar{\gamma}(\mP\vb-\mA_1\vz)\in\bF_2^m$.
\end{itemize}

We will prove that $\mQ\vx_\gamma$ is a feasible solution for $I$.
First, recall that each solution to $\mA\vx=\vb$ is of the form
$\mX\vb$, where $\mX$ is a GI of $\mA$
(Theorem \ref{CWGI1-T01}). Then, $\mX$ may be chosen of the form
$\mX=\mQ\begin{pmatrix} \mX_1 \\ \mX_2 \end{pmatrix}\mP$, where
$\mX_1$ and $\mX_2$ satisfy the property $\mA_1\mX_1+\mX_2=\mId_m$. 

As $\vb\not=\v0$, we may choose $\mX_1$ so that 
$\mX_1\mP\vb=\bar{\gamma}(\vz)=(\vx_\gamma)_{[1,n-m]}$ 
(see the discussion at the end of Section \ref{CWGI1.4}). 
Then, 
  $$\mX_2\mP\vb = \mP\vb-\mA_1\mX_1\mP\vb
                = \mP\vb-\mA_1\bar{\gamma}(\vz)
                = \bar{\gamma}(\mP\vb-\mA_1\vz)
                = (\vx_\gamma)_{[n-m+1,n]}.$$

Then, 
  $$\mX\vb=\mQ\begin{pmatrix} 
                     \mX_1\mP\vb \\
                     \mX_2\mP\vb                                
              \end{pmatrix}
          =\mQ\begin{pmatrix} 
                     (\vx_\gamma)_{[1,n-m]} \\ 
                     (\vx_\gamma)_{[n-m+1,n]} 
              \end{pmatrix}
          =\mQ\vx_\gamma$$
is a solution to $\mA\vx=\vb$.

So, we have defined a polynomial-time computable function $g$ that,
on an instance $I$ and a feasible solution $\gamma$ for $f(I)$,
returns a feasible solution $g(I,\gamma)=\mQ\vx$ for $I$.

To complete the proof, we have to show that $(f,g)$ is an $S$-reduction.
Directly from the construction, we obtain
  $$\mu_{CW}(I,\mQ\vx_\gamma)=\mu_{SAT}((\cV,\Sigma),\gamma)$$
and $\mu_{SAT}^*(\cV,\Sigma)\geq \mu_{CW}^*(I)$, where $\mu_{CW}$ 
and $\mu_{SAT}$ are the measure functions of MIN-CWP($\bF_2$) and 
MIN-SAT($\bF_2$), respectively. 

Let $I=(\mA,\vb)$ be an instance of MIN-CWP($\bF_2$). All feasible
solutions for $I$ are of the form $\mX\vb$, where $\mX$ is a generalized
inverse of $\mA$ as that above. 
Going in the opposite direction on the reasoning already done, we 
notice that for any feasible solution $\mX\vb$ for $I$, there is 
a feasible solution $\gamma$ for $f(I)$ so that $g(I,\gamma)=\mX\vb$. 
So, $g$ maps $Sol_{SAT}(f(I))$ onto $Sol_{CW}(I)$. Moreover, it is
straightforward to see that $g$ is also one-to-one. Therefore,
the property $\mu_{SAT}^*(\cV,\Sigma)=\mu_{CW}^*(I)$ must hold. 

The proof in the case of the MIN-SWP($\bF_2$) problem is similar to the 
previous one, but this time we will use Corollary \ref{CWGI2-C01}.
So, define the set of constraints by
\begin{equation}
\Sigma:\ \left\{
\begin{array}{lcl}
  z_1 &=& 1 \\
      &\cdots& \\
  z_{n-m} &=& 1 \\
  -\mA_1\vz &=& \bm{1} 
\end{array} 
\right.
\end{equation}
and $\vx_\gamma$ by
\begin{itemize}
\item $(\vx_\gamma)_{[1,n-m]}=\bar{\gamma}(\vz)\in\bF_2^{n-m}$;
\item $(\vx_\gamma)_{[n-m+1,n]}=\bar{\gamma}(-\mA_1\vz)\in\bF_2^{m}$.
\end{itemize}
The proof can now be accomplished using Corollary \ref{CWGI2-C01}.
\end{proof}

The proof of the theorem is constructive, providing the tightest 
reduction between problems. We can then propose the following
algorithms to solve our problems. Both of them use a MIN-SAT solver
denoted MIN-SAT\_{\scriptsize SOLVE}. 

\begin{algorithm}[!h]
    \caption{SAT solver for MIN-CWP$(\bF_2)$}
    \label{GID-CSP-Alg01} 
    \begin{algorithmic}[1]
        \Function{MIN-CWP\_solve}{$\mA,\vb$} 
        \State Choose a transformation 
               $\mP\mA\mQ=\begin{pmatrix} \mA_1 & \mId_m \end{pmatrix}$;
        \State Choose a vector $\vz$ of $n-m$ variables;
        \State Compute de MIN-SAT(affine) instance 
               $\Sigma=\left\{
                   \begin{array}{lcl}
                      \vz &=& {\bm{1}} \\
                      \mP\vb - \mA_1\vz &=& {\bm{1}};
                   \end{array}
                       \right.$
        \State $\gamma:=$MIN-SAT\_{\scriptsize SOLVE}$(\vz,\Sigma)$;
        \State Compute $(\vx_\gamma)_{[1,n-m]}:=\bar{\gamma}(\vz)$
               and $(\vx_\gamma)_{[n-m+1,n]}:=\bar{\gamma}(\mP\vb-\mA_1\vz)$;
        \State\textbf{return} $\mQ\vx_\gamma$. 
        \EndFunction
    \end{algorithmic}
\end{algorithm}

\begin{algorithm}[!h]
    \caption{SAT solver for MIN-SWP$(\bF_2)$}
    \label{GID-CSP-Alg02} 
    \begin{algorithmic}[1]
        \Function{MIN-SWP\_solve}{$\mA$} 
        \State Choose a transformation 
               $\mP\mA\mQ=\begin{pmatrix} \mA_1 & \mId_m \end{pmatrix}$;
        \State Choose a vector $\vz$ of $n-m$ variables;
        \State Compute de MIN-SAT(affine) instance 
               $\Sigma=\left\{
                   \begin{array}{lcl}
                      \vz &=& {\bm{1}} \\
                      - \mA_1\vz &=& {\bm{1}};
                   \end{array}
                       \right.$
        \State $\gamma:=$MIN-SAT\_{\scriptsize SOLVE}$(\vz,\Sigma)$;
        \State Compute $(\vx_\gamma)_{[1,n-m]}:=\bar{\gamma}(\vz)$
               and $(\vx_\gamma)_{[n-m+1,n]}:=\bar{\gamma}(-\mA_1\vz)$;
        \State\textbf{return} $\mQ\vx_\gamma$. 
        \EndFunction
    \end{algorithmic}
\end{algorithm}

The solver for MIN-SWP can be thought of as a particular case of 
the solver for MIN-CWP, namely the case where $\vb=\v0$. That is not 
true for SWP and CWP solvers (see Algorithms \ref{CWGI1-Alg01}
and \ref{CWGI2-Alg02}).

As we said, the reduction of MIN-CWP and MIN-SWP to MIN-SAT is very 
close and probably optimal. Only one Gaussian elimination is needed, 
after which a constraint is associated for each component of the 
solution vector. Moreover, the first n-m constraints are costless, 
of the form $z=1$. The other constraints are obtained by simply 
multiplying a matrix with a vector of Boolean variables. The solution 
returned by the result provided by the MIN-SAT solver is just as 
simple. The first $n-m$ components of the solution vector are exactly
the assignment of the Boolean variables provided by the MIN-SAT solver. 
The following components are obtained by multiplying a matrix by 
the assigned variable vector.

Currently, considerable effort is being made to streamline MIN-SAT 
solvers. They are widely used in optimization problems and have 
significant advantages over MAX-SAT solvers 
\cite{ALMZ2013}. 
These advantages are not the result of encodings but the result 
of the techniques used in MIN-SAT, some of which are not valid for 
MAX-SAT.

\section{Experimental results}\label{ExpRez}

We have implemented several GID algorithms in order to verify, test, 
and confront our theoretical results. The implementation was done in 
\texttt{Magma V2.25-3} running on a regular laptop equipped with  
Intel(R) Xeon(R) E-2176M processor with CPU @ 2.70GHz. 

\paragraph{Experiment 1}

This is an experiment for SDP. The instances are 
$(\mH,\vs,t)$, where $\mH\in\cM_{n-k,n}(\bF_q)$ and $\vs\in\bF_q^{n-k}$.
We have chosen $n=500$, $k=250$, and $q\in\{2,3,5,7\}$. 
The steps of the experiment are:
\begin{enumerate}
\item We have generated exactly one Prange transformation
      $\mP\mH\mQ=\begin{pmatrix} \mV & \mId_r \end{pmatrix}$;
\item We have performed ten iterations. They are represented, in a
      cumulative way, on the abscissa of the graphs in 
      Figure \ref{ExpRez-Fig01};
\item For each iteration, $k$ GIs were generated
      $\mX=\mQ\begin{pmatrix} \mX_1 \\ \mId_r-\mV\mX_1 \end{pmatrix}\mP$ 
      so that $\Hw{\mX_1\bar{\vs}}$ takes all the values in 
      $\{1,\ldots,k\}$; 
\end{enumerate}

Step 3 when implemented in practice in does the following. It selects a random subset $E$ of $\supp{\bar{\vs}}$ of cardinality $i\leq k.$ It computes the second part of the solution $\vx_2=\bar{\vs}-\sum_{i\in E}\mV(,i)r_i$ where $r_i$ is a random non-zero element from $\bF_q.$ It stores the weight of the solution $i+\Hw{\vx_2}.$

\begin{figure}[!ht]
 \captionsetup[subfigure]{justification=centering,singlelinecheck=false}
    \centering
    \begin{subfigure}[b]{.45\textwidth}
 {\includegraphics[width=.9\textwidth,height=.7\textwidth]{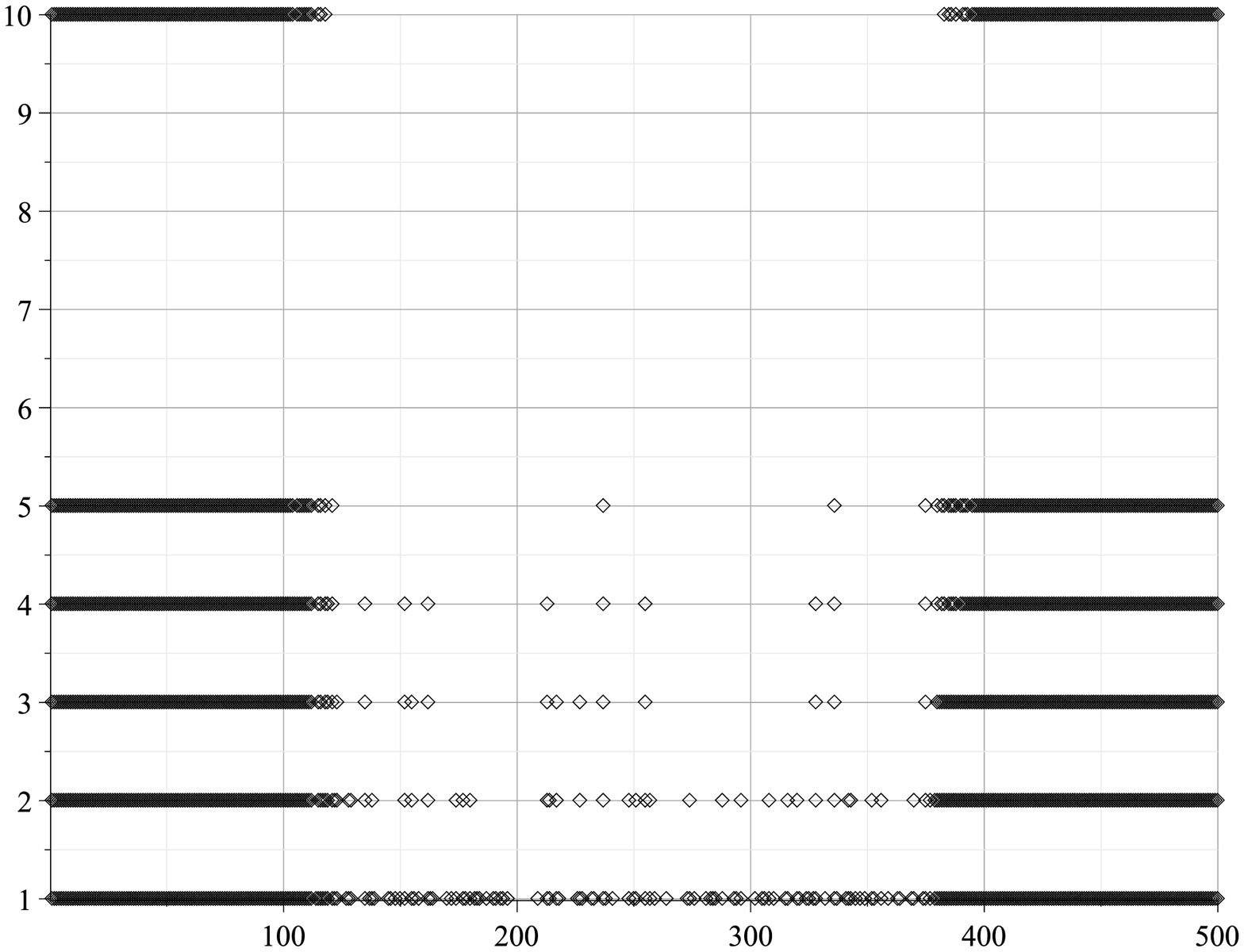}}
  \caption{$\bF_2$}
  \label{fig:W500a}
\end{subfigure}
\begin{subfigure}[b]{.45\textwidth}
    {\includegraphics[width=.9\textwidth,height=.7\textwidth]{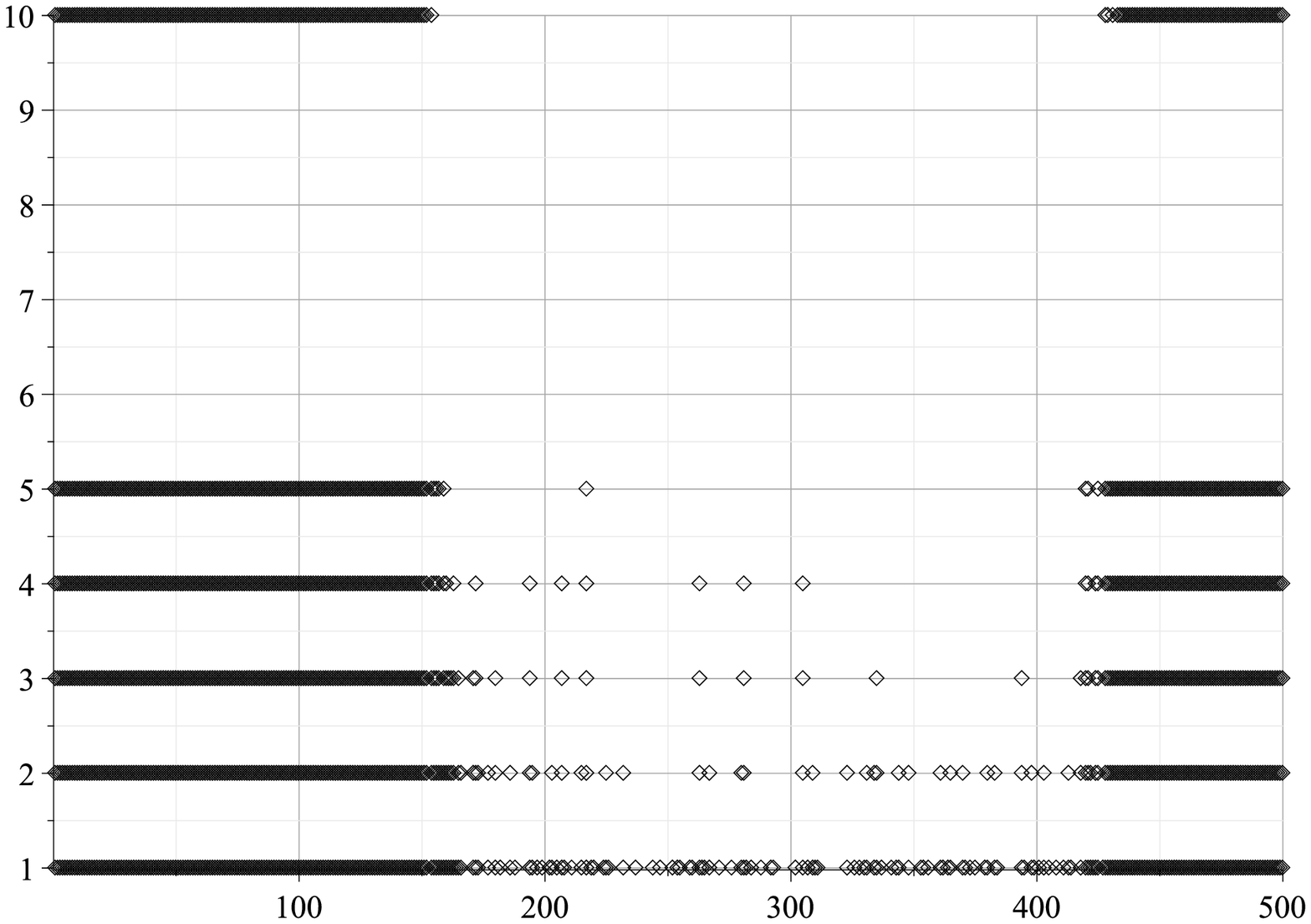}}
    \caption{$\bF_3$}
    \label{fig:W500b}
    \end{subfigure}
\begin{subfigure}[b]{.45\textwidth}
{\includegraphics[width=.9\textwidth,height=.7\textwidth]{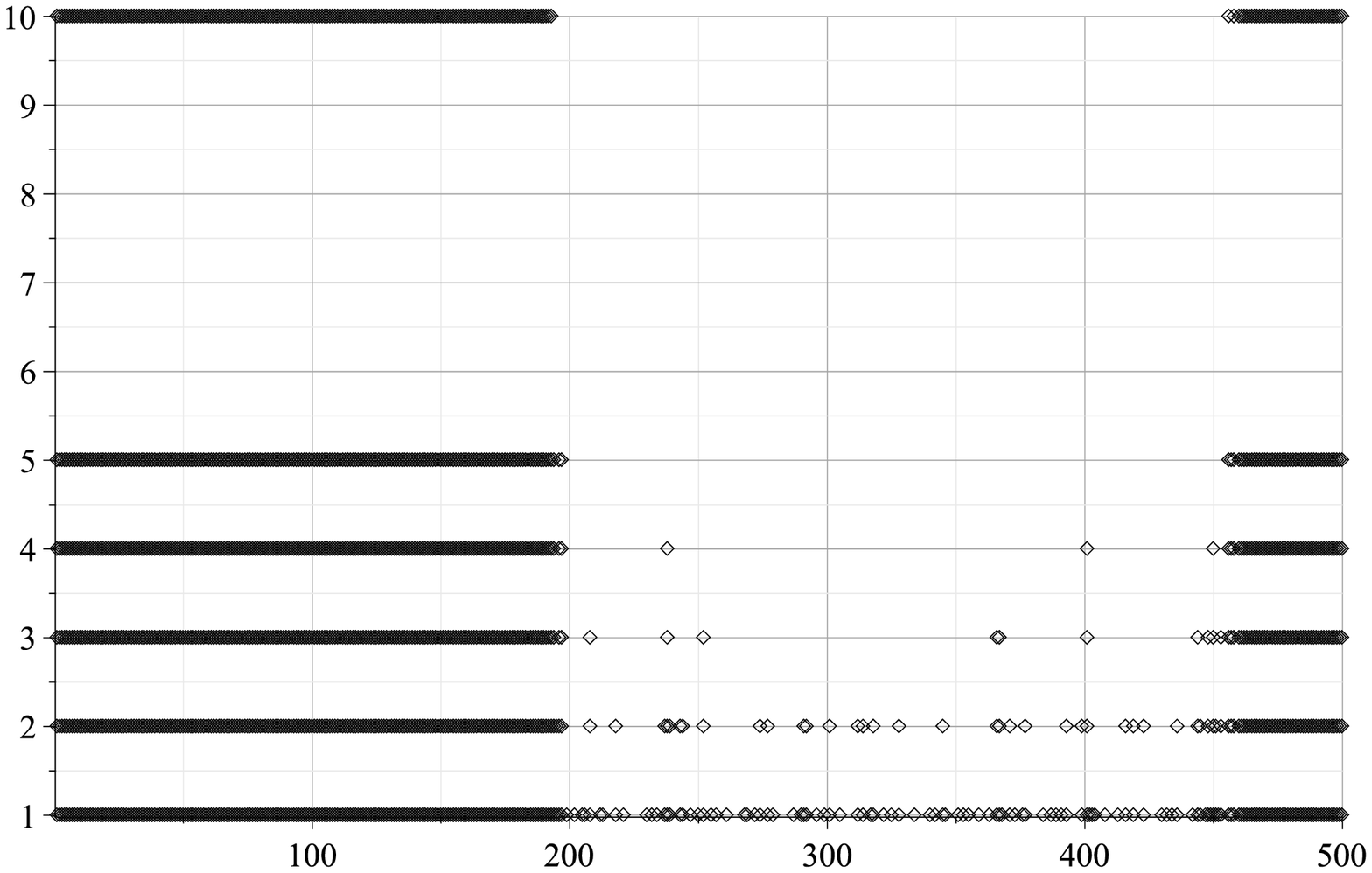}}
    \caption{$\bF_5$}
    \label{fig:W500c}
    \end{subfigure}
 \begin{subfigure}[b]{.45\textwidth}
  {\includegraphics[width=.9\textwidth,height=.7\textwidth]{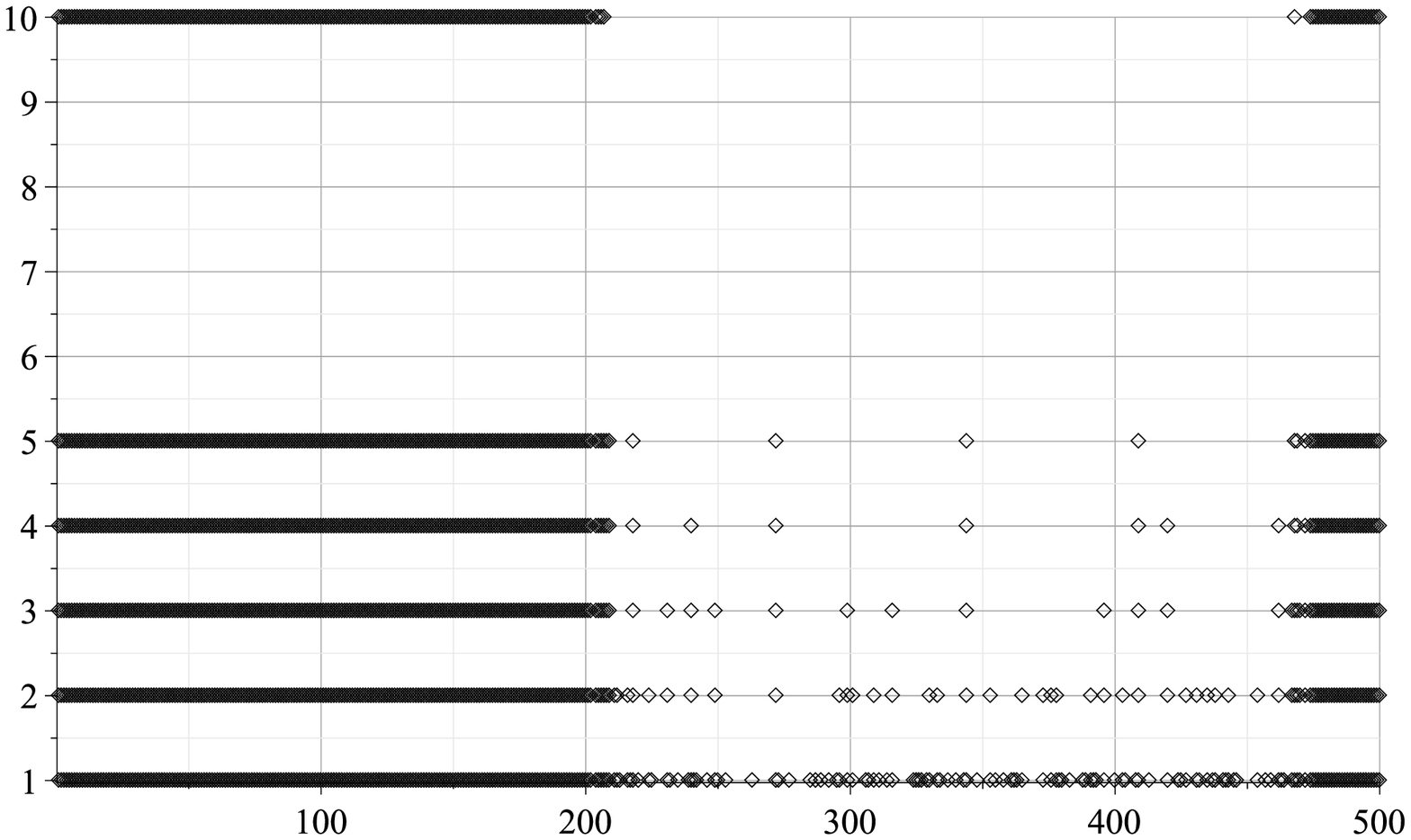}}
    \caption{$\bF_7$}
    \label{fig:W500d}
    \end{subfigure}
\caption{Hamming weights outside the range of efficiently computed solutions ($n=500,k=250$).}
\label{ExpRez-Fig01}
\end{figure}

The experiment took less than one second and the results can be seen in 
Figure \ref{ExpRez-Fig01}. 
It is noteworthy that only one Prange transformation with a very small
number of iterations produces solutions whose Hamming weight cover 
a rich range 
  $$\left[r\frac{q-1}{q},r\frac{q-1}{q}+n-r\right]$$ 
(the missing part of the line next to the $i$ value on the abscissa).

Our experiment confirms the {\em easy domain}
$[\omega_{\text{easy}}^-,\omega_{\text{easy}}^+]$ established in
\cite{DAST2019}. 
We would like to stress that we have only used one Prange transformation.

\begin{figure}[!ht]
\captionsetup[subfigure]{justification=centering,singlelinecheck=false}
    \centering
    \begin{subfigure}[b]{.45\textwidth}
{\includegraphics[width=.9\textwidth,height=.7\textwidth]{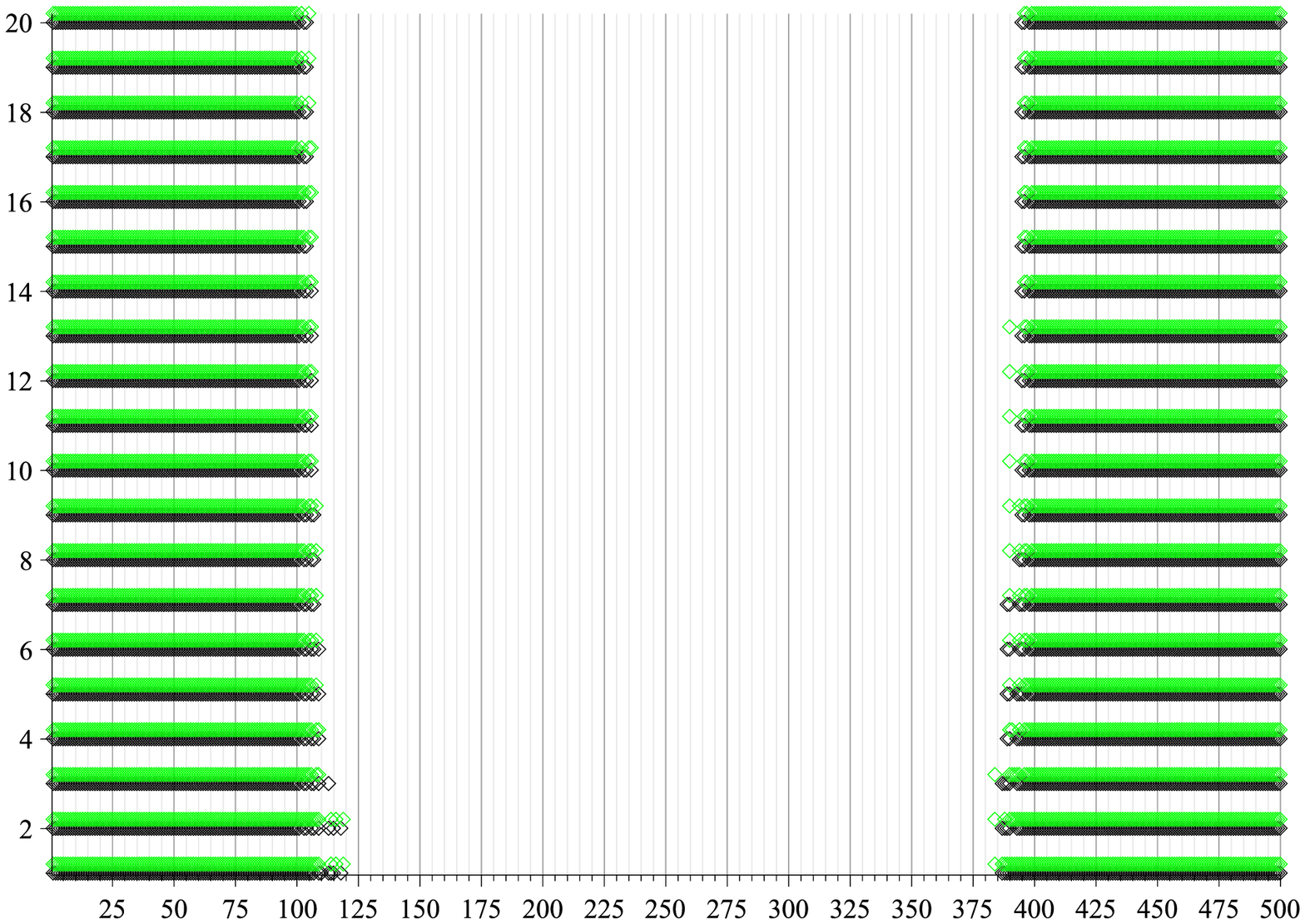}}
 \caption{$\bF_2$}
 \label{fig:W500-20PQa}
\end{subfigure}
\begin{subfigure}[b]{.45\textwidth}
    {\includegraphics[width=.9\textwidth,height=.7\textwidth]{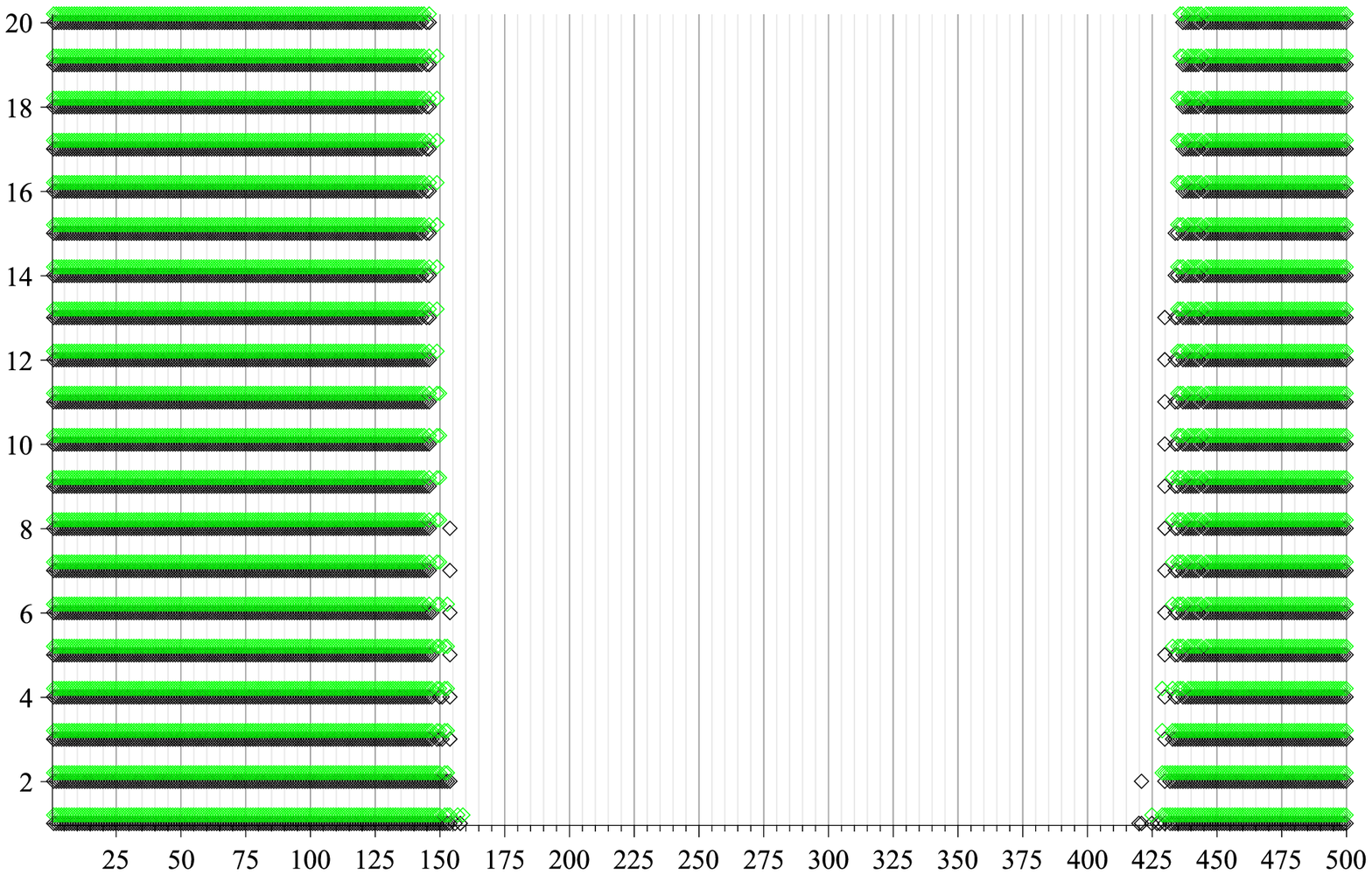}}
    \caption{$\bF_3$}
    \label{fig:W500-20PQb}
    \end{subfigure}
\caption{Hamming weights outside the range of efficiently computed solutions ($n=500,k=250$). In green a single $(\mP,\mQ)$ decomposition is used and in black 20 different decompositions are used.    \label{fig:W500-20PQ}}
\end{figure}

In order to see how the interval of easy weights behaves when multiple decompositions are used, we have repeated \textit{Experiment 1} for 20 different decompositions. We have stored an updated list of solution weights that are not retrieved and plot it in Figure \ref{fig:W500-20PQ} (in green). The same experiement was peformed using a single decomposition and plotted (in black) in the same figure. We notice the following facts.
\begin{itemize}
    \item For a fixed number of iterations, the algorithm using multiple decompositions tend to produce the same interval of weights as the algorithm using a single decomposition.
    \item For $n=500,k=250$ the value where we move from single solution to multiple solutions corresponds to $t=57$ for $\bF_2$ and $t=123$ for $\bF_3.$ Our algorithm finds solutions with weight within the interval $[102,397]$ for $\bF_2$ and $[146, 437]$ for $\bF_3.$ This corresponds to a larger interval than what in given in \cite{DAST2019}, i.e., $\left[r\frac{q-1}{q}-\sqrt{n},r\frac{q-1}{q}+n-r+\sqrt{n}\right]$, only using a small number of iterations. We have verified our statement for other code length values. For example for $n=1000,k=500$ over $\bF_3$ the value of $t=242$ and our simulations produced the interval $[300,862]$ while $\left[r\frac{q-1}{q}-\sqrt{n},r\frac{q-1}{q}+n-r+\sqrt{n}\right]=[301,801].$ 
\end{itemize}

We also wanted to illustrate how the easy weights interval behaves for other code dimensions, fact that we plot in Figure \ref{fig:W2000} for codes of length $n=2000$ and dimension increasing from $n-k=500$ to $n-k=1000.$ We stress out that even for larger values of $n$ the algorithm is fast, in less than 30 seconds we obtained the results for all the code dimensions.   

\begin{figure}[!ht]
    \centering
    \includegraphics[width=0.5\textwidth,height=0.3\textwidth]{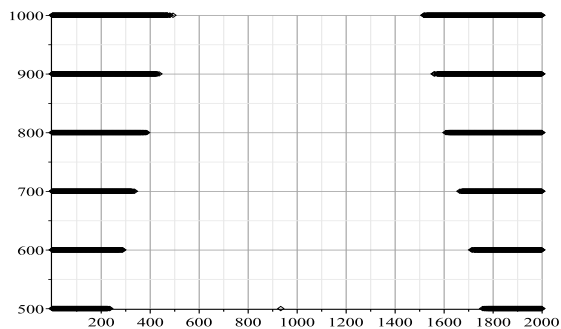}
    \caption{Hamming weights outside the range of efficiently computed solution at 10 iterations for $n=2000$ and $n-k\in\{500,600,700,800,900,1000\}.$}
    \label{fig:W2000}
\end{figure}

\paragraph{Experiment 2}

There are many variations of LWP \cite{Vard1997}. 
One of them requires finding a codeword whose weight should be in some 
given interval $[w_1,,w_2]$ (the code is specified by a parity check 
matrix). This problem is NP-complete \cite{NtHa1981}.
However, it has many easy instances as our results show. 
Moreover, an instance's hardness could depends on the positioning of 
the interval $ (w_1, w_2) $. Our simulations have shown that, as in the 
case of SDP, the interval of weights that are reached using a small 
number of iterations tends to 
$\left[r\frac{q-1}{q}-\sqrt{n},r\frac{q-1}{q}+n-r+\sqrt{n}\right].$ 
Our simulations should be interpreted as follows: 
\begin{itemize}
    \item Regarding the intractability of SDP and LWCP, as simulations 
    point our, we have good reasons to believe that in average both these 
    problems are difficult when the parameter $t$ is in a vicinity of the 
    Gilbert-Varshamov bound. 
    \item Regarding the intractability of the half-length weight codeword 
    problem \cite{DiGr1985}, our simulations show that this might not 
    be a difficult problem in average. Indeed, we were not able to find 
    codes for which the GID could not find a solution of weight $n/2$ to 
    the system $\mH\vx=\vs.$  
\end{itemize}

\section*{Acknowledgments}
This work was supported by a grant of the Ministry of Research, Innovation and
Digitization, CNCS/CCCDI – UEFISCDI,
project number
PN-III-P1-1.1-PD-2019-0285, within
PNCDI III.




\end{document}